\documentclass[twoside]{article}

\usepackage[margin=1in]{geometry}

\title{Data Assimilation in Large-Prandtl Rayleigh-B\'enard Convection
  from Thermal Measurements}
\author{A. Farhat, N. E. Glatt-Holtz, V. R. Martinez, S. A. McQuarrie, 
  and J. P. Whitehead\\
  \scriptsize{emails: afarhat@fsu.edu, negh@tulane.edu, 
    vrmartinez@hunter.cuny.edu,}\\
  \scriptsize{shanemcq@utexas.edu, whitehead@mathematics.byu.edu} }
\date{}

\usepackage{setspace} 
\usepackage{enumerate}
\usepackage{graphics} 
\usepackage{epsfig} 
\usepackage{stackrel} 
\usepackage{marginnote}

\usepackage{amsfonts, amssymb, amsmath, amsthm, multicol}
\usepackage[colorlinks=true, pdfstartview=FitV, linkcolor=blue,
            citecolor=blue, urlcolor=blue]{hyperref}
\usepackage[usenames]{color}
\definecolor{Red}{rgb}{0.7,0,0.1}
\definecolor{Green}{rgb}{0,0.7,0}
\usepackage{accents}
\usepackage{comment}
\usepackage{float}
\usepackage{graphicx}
\usepackage{multirow}
\usepackage{parcolumns}
\usepackage{subcaption}
\usepackage{textcomp}
\usepackage{mathrsfs}
\usepackage{enumerate}

\numberwithin{equation}{section}

\newcommand{\til}[1]{{\tilde{#1}}}
\newcommand{\Sob}[2]{\lVert#1\rVert_{#2}}

\newtheorem{Thm}{Theorem}[section]
\newtheorem{Lem}{Lemma}[section]

\newtheorem{Rmk}{Remark}[section]

\newtheorem*{Thm*}{Theorem}
\theoremstyle{remark}

\usepackage{cjhebrew}

\newcommand{\imb}{\hookrightarrow}

\newcommand{\lam}{\lambda}

\DeclareMathOperator{\raN}{Ra}
\DeclareMathOperator{\prN}{Pr}

\newcommand{\lb}{\langle}
\newcommand{\rb}{\rangle}

\newcommand{\bfU}{\mathbf{u}}

\newcommand{\bfW}{\mathbf{w}}

\newcommand{\bfX}{\mathbf{x}}
\newcommand{\T}{\tilde{T}}
\newcommand{\bfUT}{\tilde{\mathbf{u}}}

\newcommand{\bff}{\mathbf{f}}

\newcommand{\ddt}{\frac{d}{dt}}

\newcommand{\bfe}{\mathbf{e}}
\newcommand{\bfx}{\mathbf{x}}

\newcommand{\RR}{\mathbb{R}}

\newcommand{\Om}{\Omega}
\newcommand{\tht}{\theta}
\newcommand{\kap}{\kappa}
\newcommand{\De}{\Delta}

\newcommand{\del}{\nabla}
\newcommand{\bdy}{\partial}

\newcommand{\req}[1]{\eqref{#1}}

\newcommand{\od}[2]{\frac{d #1}{d #2}}

\begin{document}
\markboth{}{}

\maketitle

\begin{abstract}

  This work applies a continuous data assimilation scheme---a
  particular framework for reconciling sparse and potentially noisy
  observations to a mathematical model---to Rayleigh-B\'enard
  convection at infinite or large Prandtl numbers using only the
  temperature field as observables.  These Prandtl numbers are applicable to
  the earth's mantle and to gases under high pressure.  We rigorously
  identify conditions that guarantee synchronization between the
  observed system and the model, then confirm the applicability of
  these results via numerical simulations. Our numerical experiments
  show that the analytically derived conditions for synchronization
  are far from sharp; that is, synchronization often occurs even when
  the conditions of our theorems are not met.  We also develop
  estimates on the convergence of an infinite Prandtl model to a large
  (but finite) Prandtl number generated set of observations.  Numerical
  simulations in this hybrid setting indicate that the
  mathematically rigorous results are accurate, but of
  practical interest only for extremely large Prandtl numbers.
\end{abstract}

{\noindent \small
   }

\setcounter{tocdepth}{1}
\tableofcontents

\newpage

\section{Introduction} 

In order to make accurate predictions, numerical models for geophysical processes require establishing accurate initial conditions. 
Data assimilation is used to estimate weather or ocean (or any other geophysical) variables by incorporating the real world data into the mathematical system to obtain an accurate initialization. One of the classical methods of data assimilation, see, e.g., \cite{HokeAnthes1976, Daley1991, LorencBellMacpherson1991, LyneSwinbankBirch1982, VerronHolland1989, StaufferSeaman1990, ZouNavonDimet1992, StaufferBao1993, VidaradDimetPiacentini2003, AurouxBlum2008, BlomkerLawStuartZygalakis13, AzouaniOlsonTiti14}, is to insert observational measurements directly into a model as the latter is being integrated in time (also known as nudging or newtonian relaxation). There is a significant amount of recent
literature concerning the mathematically rigorous analysis of nudging
algorithms for data assimilation developed for hydrodynamic equations
with a particular focus on weather and climate systems. Recently, a
nudging scheme, known as 3DVAR, was studied in
\cite{BlomkerLawStuartZygalakis13} in the case where observables are
given as noisy Fourier modes, and in \cite{AzouaniOlsonTiti14}, which
successfully accommodates a larger class of observables that, in
particular, includes the more physically relevant cases of nodal
values and volume elements; see also \cite{BessaihOlsonTiti15}, where
observational error is accounted for.  In these articles, rigorous
proofs are obtained for the synchronization of the approximating
signal with the true signal that corresponds to the observations,
using the two-dimensional (2D) incompressible Navier-Stokes equations
(NSE) as a paradigm.

The data assimilation algorithm analyzed in \cite{BlomkerLawStuartZygalakis13,
  AzouaniOlsonTiti14} can be described as follows: suppose that $u(t)$
represents a solution of some dynamical system governed by an
evolution equation of the type
\begin{align}\label{dissipative}
\od{u}{t} = F(u),
\end{align}
where the initial state of the system, $u(0)= u_{0}$, \textit{is
  unknown}.  We would like to accurately track this solution $u(t)$ as
$t$ increases notwithstanding our uncertainty in $u_0$.  Let
$I_h(u(t))$ represent an interpolant operator based on the
observations of the system at a coarse spatial resolution of size $h$,
for $t\in [ 0,T ]$. We then construct a solution $v(t)$ from the
observations that satisfies the equations
\begin{subequations}\label{du}
\begin{align}
&\od{v}{t} = F(v) - \mu (I_h(v)- I_h(u)), \\
&v(0)= v_{0},
\end{align}
\end{subequations}
where $\mu>0$ is a relaxation (nudging) parameter and $v_{0}$ can be
prescribed as an arbitrary initial condition.  We then take $v(t)$ as
prediction of $u(t)$ which we anticipate becomes more accurate as $t$ (and
therefore the amount of observed data $I_h(u(t))$) increases.

The algorithm designated by \eqref{du} was designed to work for
dissipative dynamical systems of the form \eqref{dissipative} that are
known to have global-in-time solutions, a finite-dimensional global
attractor, as well as a finite set of determining parameters (see,
e.g., \cite{FoiasProdi1967, Foias_Temam, Foias_Titi, JonesTiti92a,
  JonesTiti92b, CockburnJonesTiti97, Holst_Titi, F_M_R_T} and
references therein).  Typically in these settings, following the ideas
in \cite{FoiasProdi1967}, lower bounds on $\mu>0$ and upper bounds on
$h>0$ can be derived such that the approximate solution $v(t)$
converges to the reference solution $u(t)$ as $t \to \infty$.  This
was initially demonstrated for the 2D Navier-Stokes equations in
\cite{BlomkerLawStuartZygalakis13, AzouaniOlsonTiti14}.

Numerous further studies, both analytical and numerical, have
been carried out for the algorithm \eqref{du}, illustrating its broad
scope of applicability.  For instance, the nudging approach has been
validated for models including the 2D magnetohydrodynamic system
\cite{BiswasHudsonLariosPei17}, the 2D surface quasi-geostrophic
equation \cite{JollyMartinezTiti17}, three-dimensional (3D)
Brinkman-Forchheimer-Extended Darcy model
\cite{MarkowichTitiTrabelsi16}, and 3D simplified Bardina model
\cite{AlbanezBenvenutti18}.  The practically and physically relevant
scenarios of discrete-time and time-averaged observables was studied
in \cite{FoiasMondainiTiti16, BlocherMartinezOlson18,
  JollyMartinezOlsonTiti18}; more recently, it was shown in
\cite{BiswasFoiasMondainiTiti18} that this nudging algorithm is
capable of synchronizing the statistics propagated by the flow as they
are observed only on a coarse-mesh scale; the efficacy of this
algorithm for assimilating actual data sampled from a regional domain
encompassing most of Northern Africa and the Middle East was recently
tested in \cite{DesamsettiDasariLangodanTitiKnioHoteit18}.  We refer
the reader to \cite{Daley1991} for a summary on the use of data
assimilation in practical forecasting and \cite{Kalnay03} for a
comprehensive text on numerical weather prediction where nudging 
has been employed. 

Regarding related numerical studies, \cite{GeshoOlsonTiti}
demonstrated in the case of the 2D NSE that the number of observables
required for synchronization using \eqref{du} is much lower in
practice than what has been deemed sufficient by the rigorous
analysis.  In the setting of the 2D RB system, numerical studies were
carried out in \cite{AltafTitiGebraelKnioZhaoMcCabeHoteit15}, and then
in \cite{FarhatJohnstonJollyTiti18} for nearly turbulent flows using
vorticity and local circulation measurements.  We emphasize that the
numerical experiments carried out in the present article are for
moderately turbulent flows whose dynamics are significantly more
complex than the regime of two-cell convection rolls that
\cite{AltafTitiGebraelKnioZhaoMcCabeHoteit15} was restricted to.
Moreover our studies are carried out to a similar high degree of
numerical precision as found in \cite{FarhatJohnstonJollyTiti18}.  We
refer the reader to \cite{AlbanezLopesTiti16, FarhatLunasinTiti17,
  LeoniMazzinoBiferale18, CarlsonHudsonLarios18} for various other
studies in the context of turbulent flows such as how one can leverage
the nudging scheme to infer unknown parameters of the flow.  In
\cite{BiswasMartinez17, IbdahMondainiTiti18, MondainiTiti18}
analytical studies on the various modes of synchronization of the
algorithm \eqref{du} and on certain variants on its numerical
discretization were carried out.

The earth system is heated from within and cooled by the atmosphere or ocean at the earth's surface. 
On geological time scales the mantle's motion can be modeled as a fluid. 
The big difference between the temperature of the top mantle and the bottom mantle 
is a major source of the convective motion (fluid motion driven by temperature difference). The full compressible,
temperature-dependent viscous equations of motion that ostensibly
describe flow in the mantle \cite{turcotteschubert2014} are currently
beyond the reach of a rigorous mathematical analysis, but
a first order approximation to this system is adequately described by
taking the infinite Prandtl limit
\cite{Wang2004a,Wang2004b,XWang05,Wang2007,Wang2008a,Wang2008b,foldesglattholtzrichards2015}
of the Rayleigh-B\'enard (RB) system first described in \cite{Ra1916}
by Lord Rayleigh. We recall that the Prandtl number represents
the ratio of the kinematic viscosity to the thermal diffusivity.
Since the original formulation of the minimal mathematical model in
\cite{Ra1916}, extensive research has sought to quantify its dynamical
evolution, \cite{Lo1963,Ah1974,AhBe1978,BoBrAhCa1991}, and resultant
large spatial and long temporal scale impact of convective flow, see
\cite{AhlersGrossmannLohse2009,LoXi2010} for example.  Despite the
seeming simplicity of the RB system, there remains open questions
regarding the exact nature of the convective heat transport, and the
impact and nature of boundary layers at physically relevant values
(see e.g. \cite{AhlersGrossmannLohse2009}).  To further complicate
matters, mantle convection is far more nuanced than Rayleigh-B\'enard
convection, having several other unanswered questions, in addition to
the well-known open problems in the latter setting.  Unlike the low
Prandtl number setting, experimental investigations of mantle
convection are not practical, so numerical simulations provide one of
the only avenues to investigate these issues.  Of fundamental concern
in such simulations is the dependence of the simulation on the initial
condition and/or true physical setting, which would ideally be
accompanied by physical observations. The collection of
observational data from the mantle is an onerous inverse problem that
obfuscates much of the desired resolution in time and space
\cite{turcotteschubert2014}, both in the sense of physical
accessibility, and due to the relevant time scales involved.  Indeed,
the evolution of the mantle is on millenial time scales.  Despite
remarkable advances in imaging technologies, observations of the mantle are sparse and prone to
noise and are insufficient to determine the mantle's state. Thus, 
the development of the advanced real-time prediction systems that are capable of 
depicting and predicting the mantle's state is necessary to gain insight into the dynamics in the earth's interior.

Access to observations from earth's mantle is limited. Geophysicists have decent observations at the surface (top layer) of plate velocities and heat flow distribution and have probes of mantle temperature where volcanism occurs. Since the motion of the tectonic plates is very slow (the relative movement of the plates typically ranges from zero to 100 mm annually), we can assume that the top plate is stationary and with fixed temperature.
Thus, a realistic forecast model for the dynamics of earth's mantle will employ {\it sparse thermal observations only}. In the context of atmospheric and oceanic physics, data assimilation algorithms where some state variable observations are not available as an input, have been studied in \cite{CharneyHalemJastrow69, errico1987predictability, Ghil1977, Ghil-Halem-Atlas1978} for simplified numerical forecast models. Particularly related to the study carried out in this article, it was
shown in the case of the 2D NSE \cite{FarhatLunasinTiti16a} and the 2D
RB system in \cite{FarhatLunasinTiti16d} that measurements on just a
single component of velocity is sufficient to obtain synchronization.
 Charney's question in \cite{CharneyHalemJastrow69, Ghil-Halem-Atlas1978, Ghil1977} asks whether temperature observations are enough to determine the entire dynamical state of the system. In \cite{Ghil1977}, an analytical argument suggested that Charney's conjecture is correct, in particular, for a shallow water model. Further numerical testing in \cite{Ghil-Halem-Atlas1978} affirmed that it is not certain whether assimilation with temperature data alone will yield initial states of arbitrary accuracy. The authors in \cite{AltafTitiGebraelKnioZhaoMcCabeHoteit15} concluded that assimilation using coarse temperature measurements only will not always recover the true state of the full system. It was observed that the convergence to the true state using temperature measurements only is actually sensitive to the amount of noise in the measured data as well as to the spacing (the sparsity of the collected data) and the time-frequency of such measured temperature data.
Rigorous justification for Charney's Conjecture was provided in
\cite{FarhatLunasinTiti16c} in the case of the 3D Planetary
Geostrophic model.  Earlier, for the specific setting of 3D convection
in a porous medium, where inertial effects can be ignored in the fluid
velocity, it was shown in \cite{FarhatLunasinTiti16b} that temperature
measurements alone suffice to determine the velocity field.  By
comparison, the thrust of the analysis performed here is to establish
the conclusions analogous to \cite{FarhatLunasinTiti16b} while accounting
for these inertial effects within the regime of a finite but large Prandtl
number.

We consider the nudging
approach both analytically and through numerical experiments to
explore the range of applicability of the technique in this
geophysically interesting context of large Prandtl convective systems.
Ultimately, we accomplish the following:
\begin{enumerate}
\item We develop a nudging data assimilation scheme for both large and
  infinite Prandtl number Rayleigh-B\'enard convection in the
  traditional simplified three-dimensional box geometry
  (cf. \eqref{nd:Bou},\eqref{nd:bc} and \eqref{system:data-infinite} below)
  with observations in the temperature field only. In
  Section~\ref{sec:RB} we formally introduce the governing equations for
  the 3D RB system and then in Section~ \ref{sec:preliminaries}, we
  provide the mathematical framework within which our analysis is
  performed, as well as the relevant well-posedness results for the 3D
  RB system.  We then establish rigorous estimates on the convergence
  rates for the simpler case of $\Pr=\infty$ in Section
  \ref{section:infinite-prandtl-assimilation}.  The case of
  large, but finite $\Pr$ is addressed in Section
  \ref{section:finite-prandtl-assimilation}.

\item Perform high-resolution direct numerical simulations (DNS) on
  the \textit{two-dimensional} version of this problem for moderately
  turbulent flows, in an effort to shed some light on the practical
  applicability of the rigorous estimates. In particular, we probe the
  values of the relaxation parameter and the number of required modes
  for the nudging scheme to converge.  This is done in Sections
  \ref{section:numerical-results-infinite} and
  \ref{section:numerical-results-finite}, immediately after their
  respective mathematical analysis.

\item Consider a practical scenario of `model error', in which the
  assimilated variables are nudged ``incorrectly." Specifically, we
  assume that the modeling system corresponds to the infinite Prandtl
  system nudged by \textit{data corresponding to a finite Prandlt
    system} \req{nd:Bou}, \req{nd:bc}. This situation is studied both
  analytically and numerically in Section
  \ref{section:hybrid-prandtl-assimilation}.

\end{enumerate}
We note that the choice of finitely many Fourier modes as the
manifestation of our observables is made for ease of both exposition
and numerical implementation. We reserve establishing estimates on the
convergence rates for more general observables to a subsequent study.

\section{The Rayleigh-B\'enard System and Nudging Equation}
\label{sec:RB}

This initial section recalls the Rayleigh-B\'enard (RB) system and its
non-dimensional formulation.  We then present the precise form of the
nudging algorithm which we will study in the sequel.

The Rayleigh-B\'enard system for convection originates from the
Boussinesq equations for an incompressible fluid with appropriate
boundary conditions. The Boussinesq system over a $d$-dimensional
domain, where $d=2,3$, $\Omega = [0,\tilde{L}]^{d-1} \times [0,h]$, is given by
\begin{align}\label{eq:Bou}
\partial_t \bfU + (\bfU \cdot \nabla) \bfU = \nu\Delta \bfU 
-  \nabla p  + \alpha g \bfe_d T , \quad \nabla \cdot \bfU = 0,
  \quad \partial_t T +( \bfU \cdot \nabla) T - \kap\Delta T=0,
\end{align}
where $\bfU = (u_1,\dots,u_d)$ is the velocity vector field, $p$ is the
scalar pressure field, and $T$ denotes the temperature of a buoyancy
driven fluid. The parameter $\nu$ denotes the kinematic viscosity of
the fluid, $\kap$ its thermal diffusivity, $\alpha$ the thermal
expansion coefficient, $g$ denotes the constant gravitational force
and $\bfe_d$ is a constant vector anti-parallel to the gravitational force. To model convection,
\eqref{eq:Bou} is then supplemented by the boundary conditions
\begin{align}   \label{eq:bc}
   \bfU|_{x_d = 0} = \bfU|_{x_d = h} = \mathbf{0}, 
  \quad 
  T|_{x_d = 0} = \delta T, \;  T|_{x_d = h} = 0, 
  \quad
  \bfU, T \textrm{ are $L$-periodic in }x_1,x_{d-1},
\end{align}
where $\delta T$ is a fixed constant determined by the (relative)
strength of the bottom heating. The relations \req{eq:Bou},
\req{eq:bc} together constitute the 3D RB system for convection.  We
note that variations on these boundary conditions are applicable to
the earth's mantle, but as our focus is on the convergence of the data
assimilation scheme, we assume that such variations have secondary
effects.

\subsubsection*{Non-dimensionalized variables}

As is customary, we work with non-dimensionalized variables.  The
system \eqref{eq:Bou} is rescaled using $h$ as a length scale,
$\delta T$ as the temperature scale, and the diffusive scale
$\frac{\kappa}{h^2}$ as the time scale.  The relevant non-dimensional
physical parameters for the system are the Prandtl number, $\prN$, and
Rayleigh number, $\raN$, which are defined as
\begin{align}\label{pr:ra}
  \prN:=\frac{\nu}{\kap},
    \qquad
    \qquad
  \raN:=\frac{\alpha g(\delta T)h^3}{\nu \kap}.
\end{align}
This leads to non-dimensionalized variables over the rescaled domain
$\Om'=[0,L]^{d-1}\times[0,1]$, $d=2,3$, which satisfy
\begin{align}\label{nd:Bou}
  \begin{split}
    \frac{1}{\prN}\left[\partial_{t'} \bfU' + (\bfU' \cdot \nabla')\bfU'\right] 
    -\Delta' \bfU'
        &=  -  \nabla' p'  + \raN\bfe_d T' , 
        \quad \nabla' \cdot \bfU' = 0,
        \quad \bfU'(\bfx',0) = \bfU'_0(\bfx')\\
        \partial_{t'} T' + \bfU' \cdot \nabla'T' -\Delta' T'&= 0,
        \quad T'(\bfx',0) = T'_0(\bfx')
   \end{split}
\end{align}
with the boundary conditions
\begin{align}\label{nd:bc}
  \bfU'|_{x_d' = 0} = \bfU'|_{x_d' = 1} = \mathbf{0}, 
  \quad T'|_{x_d' = 0} = 1, \;  T'|_{x_d' = 1} = 0, 
  \quad 
  \bfU', T' \textrm{ are $L$-periodic in }x_1', x_{d-1}',
\end{align}
For notational simplicity, we will drop the $'$ in all that follows.  

As previously alluded to, the physical setting of interest in this
article is the earth's mantle where the Prandtl number is
\emph{large}, namely, on the order of $10^{25}$.  Upon formally
setting $\prN = \infty$ in the system \eqref{nd:Bou}--\eqref{nd:bc},
one arrives at
\begin{gather}
  \label{system:data-infinite}
  \begin{gathered}
    -\Delta \bfU =-\nabla p  + \raN \bfe_d T, 
    \quad \nabla \cdot \bfU = 0,\\
    \partial_t T + \bfU \cdot \nabla T - \Delta T = 0,
    \quad T(\bfX,0)=T_0(\bfX),\\
    \bfU|_{x_d = 0} = \bfU|_{x_d = 1} = \mathbf{0}, 
    \quad T|_{x_d = 0} = 1,\; T|_{x_d = 1} = 0, 
    \quad
    \bfU, T \textrm{ are $L$-periodic in }x_1, x_{d-1}.
  \end{gathered}
\end{gather} 
We note that the initial velocity $\bfU(\bfX,0)$ is then determined by
$T_0$ and the corresponding momentum equations. Although there are
several additional physical effects relevant to mantle convection that
are omitted from the Rayleigh-B\'enard model considered here, we
consider \req{system:data-infinite} to be an appropriate ``zeroth
order'' representation of mantle convection; it provides the starting
point and test model for mantle convection simulations
(cf. \cite{Blankenbachetal1989}).  Although we anticipate eventually
extending the results of the current investigation to more realistic
models of mantle convection, we believe a more in-depth understanding
of the problem at hand is a necessary first step.

\subsubsection*{Nudging setup}

The idea, following \cite{BlomkerLawStuartZygalakis13,
  AzouaniOlsonTiti14} and sketched in the introduction, is to nudge
the assimilated system as in \req{du} with a projection of the `truth'
that represents the exactly realizable observations of the original
system.  More precisely, the nudging is accomplished by introducing an
affine feedback control term to the original `forecast' model
\req{dissipative}, whose purpose is to enforce the asymptotic
convergence of the solution of the assimilated system \req{du} towards
that of the original system \req{dissipative}, \textit{but only on the
  scales at which the observations are made}; it is this `relaxed'
imposition that ensures the practicality of the nudging scheme. For
this article, our `truth' is assumed to be represented by
\req{nd:Bou}, \req{nd:bc} or by \eqref{system:data-infinite}.

Let $(\bfU,T)$ satisfy \req{nd:Bou}, \req{nd:bc} over
$\Om=[0,L]^2\times[0,1]$, from which we have obtained partial
observations in the form of finitely many Fourier coefficients
corresponding to wave-numbers $|k|\leq N$, for some integer $N>0$.  Let $(\tilde{\bfU},\T)$ denote the assimilated or modeled system
variables, which satisfies
\begin{gather}\label{eq:B:DA}
  \begin{gathered}
   \frac{1}{\prN}[\partial_{t} \bfUT + (\bfUT \cdot \nabla) \bfUT ] -\Delta \bfUT 
   =-\nabla\tilde{p}  + \raN \bfe_d \T , 
   \quad \nabla \cdot \bfUT = 0,
   \quad \bfUT(\bf{x},0)=\bfUT_0(\bfx)\\
  \partial_t \T + \bfUT \cdot \nabla \T - \Delta \T= - \mu P_N (\T - T),
  \quad \T(\bfx,0)=\T_0(\bfx),\\
   \bfUT|_{x_d = 0} = \bfUT|_{x_d = 1} = \mathbf{0}, 
   \quad \T|_{x_d = 0} = 1,\; \T|_{x_d = 1} = 0, \quad
   \bfUT, \T \textrm{ are $L$-periodic in }x_1,x_{d-1},
 \end{gathered}
\end{gather}
where $P_N$ denotes the projection onto Fourier wave-numbers
$|k|\leq N$ (see \req{projections} below). Its infinite Prandtl
counterpart is given by
\begin{gather}
  \begin{gathered} 
    -\Delta \bfUT =-\nabla\tilde{p}  + \raN \bfe_d \T, 
    \quad \nabla \cdot \bfUT = 0,\\
    \partial_t \T + \bfUT \cdot \nabla \T - \Delta \T = - \mu P_N (\T - T),
    \quad \T(\bfx,0)=\T_0(\bfx),\\
    \bfUT|_{x_d = 0} = \bfUT|_{x_d = 1} = \mathbf{0}, 
    \quad \T|_{x_d = 0} = 1,\; \T|_{x_d = 1} = 0, \quad
    \bfUT, \T\ \textrm{are $L$-periodic in }x_1,x_{d-1},
\end{gathered}\label{system:nudge-infinite}
\end{gather}
where $(\bfU,T)$ comes from either \req{nd:Bou}, \req{nd:bc} or \req{system:data-infinite}.

One of
the basic goals of this paper is to show that
$\T-T,\tilde{p}-p\rightarrow 0$ and $\bfUT-\bfU\rightarrow\mathbf{0}$ as
$t\rightarrow\infty$ in the appropriate space for specific conditions
on $\mu$ and the number of projected modes $N$ relative to $\raN$ and
$\prN$.  This indicates that for specified Rayleigh $\raN$ and Prandtl
$\prN$ numbers one can determine a sufficiently large number of modes
and a sufficiently large nudging parameter $\mu$ to ensure that the
assimilated system $(\tilde{\bfU},\tilde{T})$ will asymptotically
match the true system.

\section{Mathematical Background} 
\label{sec:preliminaries}
For the sake of completeness, this section presents some preliminary
material and notation commonly used in the mathematical study of
hydrodynamic systems, in particular in the study of the NSE
    and the Euler equations for incompressible fluids.  For more
detailed discussion on these topics, we refer the reader to, e.g.,
\cite{ConstantinFoias88,Temam1997}.

Let $\Omega = [0,L]^{d-1}\times[0,1]$, where $d=2,3$ and 
we denote the spatial variable $\bfX = (x_1, \ldots, x_d)$.
We consider the function spaces
\begin{align}
  \mathscr{F}&:=\{v\in C^\infty(\Om):L\text{-periodic in}\ x_j, j=1,d-1,\ 
               \text{compactly supported in}\ x_d\},
    \label{eq:smooth:test}\\
  \mathscr{F}^d_{\sigma}&:=\{\mathbf{v}\in \mathscr{F}^d:\del\cdotp\mathbf{v}=0\},
    \notag\\
  H&:=\overline{\mathscr{F}}^{L^2},
     \quad \mathcal{H} :=\overline{\mathscr{F}_\sigma^d}^{L^2}\label{eq:H-def}\\
  {V}&:=\overline{\mathscr{F}}^{H^1},
     \quad\mathcal{V}:=\overline{\mathscr{F}^d_{\sigma}}^{H^1}
       \label{eq:V-def}\\
  {W}&:=\overline{\mathscr{F}}^{H^2},
       \quad \mathcal{W}:=\overline{\mathscr{F}^d_{\sigma}}^{H^2}\label{eq:W-def}
\end{align}
where $H^1(\Omega)$ and $H^2(\Omega)$ denote the classical Sobolev
class of first-order and second-order weakly differentiable functions
over $\Omega$, respectively. We
use the notation $X^{\times d}$ to denote the $d$-fold product of a
set $X$, and $\lb\cdot,\cdot\rb$ to denote the usual $L^2$ inner
product over $\Omega$,
\begin{align*}
  \lb{\bf{u}}, {\bf{v}}\rb = \sum_{i=1}^3  \int_\Omega u_i(\bfX)v_i(\bfX)\:d\bfX, 
  \qquad \lb f, g\rb = \int_\Omega f(\bfX)g(\bfX)\:d\bfX,
\end{align*}
for
${\bf u}= (u_1,\dots, u_d), {\bf v}= (v_1,\dots, v_d) \in
{\mathcal{H}}$,
and $f, g\in L^2(\Om)$. The inner product on $H^k(\Om)$,
$k=1,2,\dots$, is given by
\begin{align*}
		\lb f,g\rb_{H^k}:=\sum_{|\gamma|=0}^k\lb D^\gamma f,D^\gamma g\rb,
\end{align*}
where $\gamma=(\gamma_1,\dots,\gamma_d)$ is a multi-index and
$D^\gamma=(\bdy_{x_1}^{\gamma_1},\dots, \bdy_{x_d}^{\gamma_d})$. The
spaces $H,V,W$ are then endowed with a Hilbert space structure, whose
respective inner products are given by
\begin{align}\label{inn:prod}
  (f,g)_H:=\lb f,g\rb,\quad 
  (f,g)_V:=\lb\del f,\del g\rb,\quad 
  (f,g)_W:=(f,b)_V+\sum_{|{\gamma}|=2}\lb D^\gamma f,D^\gamma g\rb.
\end{align}
The spaces $\mathcal{H},\mathcal{V},\mathcal{W}$  have analogous Hilbert
space structures. We denote by $H',V',W'$ and
$\mathcal{H}',\mathcal{V}',\mathcal{W}'$ the dual spaces of $H,V,W$,
and $\mathcal{H},\mathcal{V},\mathcal{W}$, respectively. We then have
the following continuous injections:
\begin{align}\label{embed}
  \begin{split}
    &W\imb V\imb H\imb H'\imb V'\imb W',\\
    &\mathcal{W}\imb\mathcal{V}
      \imb\mathcal{H}\imb\mathcal{H}'\imb\mathcal{V}'\imb\mathcal{W}'.
  \end{split}
\end{align}
In what follows we will denote the $L^2(\Om)$ norm by $\| \cdot \|$.
For all other Banach spaces $X$, e.g. $L^p(\Om)$, for $p \not= 2$,
$ H^k(\Om)$ etc., we denote the associated norms explicitly as
$\| \cdot \|_X$.

Let 
\begin{align}\label{eq:eigenpairs}
\{(\lambda_{n},\phi_{n}({\bf x}))\}_{n=1}^\infty 
\end{align}
denote the orthonormal eigenpairs corresponding to the Laplace
operator $-\Delta$ on the domain $\Om$ supplemented with the mixed
horizontally periodic-vertically Dirichlet boundary condition as in
\eqref{eq:smooth:test}.  Then each $f\in H^2(\Omega)$ can be expressed
in terms of the eigenfunctions as
\begin{align*}
f(\bfX,t) = \sum_{n=1}^\infty f_{n}(t)\phi_n(\bfX),
\qquad
f_{n}(t) = \lb f(t),\phi_n\rb 
= \int_\Omega f(\bfX,t)\phi_n(\bfX)\:d\bfX,
\end{align*}
and the eigenfunctions satisfy the orthogonality relation
\begin{align*}
\langle\phi_i,\phi_j\rangle 
= \delta_{ij}
= \begin{cases}
1 & \textrm{if }i = j,\\
0 & \textrm{if }i \neq j.
\end{cases}
\end{align*}
For each $N\geq0$, define the projections $P_N, Q_N$ by
\begin{align}
  P_N(f) = \sum_{n=1}^N f_n\phi_n,
  \qquad
  Q_N(f) = [I - P_N](f) = \sum_{n=N+1}^\infty f_n\phi_n,
  \label{projections}
\end{align}
where $I$ is the identity operator.  In other words, $P_N$ is a
truncation of the eigenfunction expansion, and $Q_N$ is its orthogonal
complement.  The orthogonality of the eigenbasis yields the identities
\begin{align}
\label{eq:pnqn-vanishes}
\langle P_N(f), Q_N(f) \rangle
&= 0,
\\
\|P_N(f)\|^2 + \|Q_N(f)\|^2
&= \|f\|^2,
\label{eq:pn+qn=f}
\end{align}
for any $f\in H^2(\Omega)$.  

We next recall the following well-known a priori estimates for Stokes'
equations
\begin{gather}
        -\De\bfU+\del p=\bff,
        \label{eq:stokes:eq}\\
        \del\cdotp\bfU=0,
        \label{eq:stokes:Divfree}\\
        \bfU|_{x_3=0}=\bfU|_{x_3=1}=\mathbf{0}, \quad
        \bfU\textrm{ is $L$-periodic in }x_1\textrm{ and }x_2,
        \label{eq:stokes:BC}
\end{gather}
cf. \cite{ConstantinFoias88} and the drift-diffusion equation, \eqref{T:BCs}
where the advecting velocity field is divergence free
as in e.g. \cite{Temam1997,XWang05}.  We adapt these results here to our
notation and setting as follows
\begin{Lem}
\label{lem:stokes}
Let $d=2,3$ and $\bff\in H^{\times d}$. There exists a unique
$\bfU\in\mathcal{W}$ and (up to constants) $p\in V$ such that \eqref{eq:stokes:eq} is satisfied.
Moreover, there exists a constant $C = C(\Omega) > 0$ such that
\begin{align*}
  \|\bfU\|_{H^2}+ \|p\|_{H^1} \leq C\|\bff\|,
\end{align*}
\end{Lem}

 The next lemma follows from the theory of linear transport equations in \cite{DiPerna_Lions} and a variant of the maximum principle proved in \cite{Foias_Manley_Temam_87} .We will refer to the following notation for the
``positive" and ``negative parts" of a function:
\begin{align}\label{pos-neg:part}
  \psi^+:=\max\{\psi,0\},\quad \psi^-:=\max\{-\psi,0\}.
\end{align}

\begin{Lem}
\label{lem:mp}
Let $d=2,3$ and $\tau>0$. Let $\bfU\in L^1(0,\tau;\mathcal{V})$ and  $T_0\in L^\infty(\Om)$ be {a.e.} $L$-periodic in $x_1,x_{d-1}$ and $T_0|_{x_d=0}=0,T_0|_{x_d=1}=1$ (in the sense of trace). Suppose that 
$T\in L^\infty(0,\tau;L^\infty(\Om))\cap L^2(0,\tau;H^1(\Om))$ satisfies
\begin{gather}
  \partial_t T + \bfU\cdot\del T - \Delta T = 0,
  \quad T(\bfX,0) = T_0(\bfX),\notag\\
  \qquad
  T|_{x_d=0} = 1,
  \qquad
  T|_{x_d=1} = 0,\qquad T\ \text{is $a.e.$ $L$-periodic in}\ x_1,x_{d-1},
  \label{T:BCs}
\end{gather}
where the boundary values on $\{x_d=0\}\cup\{x_d=1\}$ are interpreted
in the sense of trace. Then there exists a constant $C_0=C_0(\Om,\Sob{T_0}{})>0$ and functions $\bar{T},\eta$ such that $T=\bar{T}+\eta$ and satisfy
\begin{align*}
 	0\leq\bar{T}(t)\leq1,\quad \eta=(T-1)^+-T^-,\quad \|\eta(t)\|\leq C_0e^{-t}.
\end{align*}
for all $t>0$.
\end{Lem}

For the system \req{system:data-infinite}, when $\prN=\infty$, the
velocity field, $\bfU$, is determined by the evolution of $T$. The
well-posedness of this system then follows in a standard way, for
either dimension $d=2,3$, and its solution satisfies the estimates
stated in Lemma \ref{lem:stokes} and \ref{lem:mp}. We formally state this result as the following theorem.

\begin{Thm}\label{cor:infinite:prandtl}
  Let $d=2,3$. Suppose that $T_0\in L^\infty(\Om)$ is a.e. in $\Om$, that $T_0$ is
  a.e. $L$-periodic in $x_1,x_{d-1}$ and
  $T_0|_{x_d=0}=0,T_0|_{x_d=1}=1$ (in the sense of trace). Then there
  exists a unique $(\bfU,T)$ satisfying \req{system:data-infinite}
  such that
  \begin{align*}
    \bfU \in L^\infty(0,\tau;\mathcal{W}),
    \quad 
    T\in L^\infty(0,\tau;L^2(\Om))\cap 
         L^2(0,\tau;H^1(\Om))\cap C_w([0,\tau];L^2(\Om)),
  \end{align*}
  for all $\tau>0$. Moreover, there exist positive constants $\gamma_0=\gamma_0(\Om,\Sob{T_0}{})$ and $C_0=C_0(\Om,\|T_0\|)$, such that 
  \begin{align*}
    \raN^{-1}\Sob{\bfU(t)}{H^2}
        \leq \gamma_0,\quad\|T(t)\|\leq \gamma_0',
  \end{align*}
for all $t\geq0$, where $\gamma_0':=|\Om|^{1/2}+C_0e^{-t}$.
\end{Thm}

We consider a change of variable, denoted by $\theta({\bf x},t)$, that
represents the fluctuation of the temperature around the steady state
background temperature profile $1-x_d$:
\begin{align}\label{perturb:var}
  \theta({\bf x},t) = T({\bf x},t) - (1-x_d).
\end{align}
The functional setting determined by
\req{eq:H-def}-\req{eq:W-def} accommodates a rigorous mathematical
analysis for the perturbed variable $\tht$. The results derived for
$\tht$ are then transferred naturally to the desired results for the
original variable $T$. We appeal to \cite{XWang05, Wang2007} for the
global existence and eventual regularity of suitable weak solutions
for \req{nd:Bou}-\req{nd:bc}, as well as the existence of the ``global attractor" for the dynamics, although
we will not make explicit use of this fact in this article. 
We will also say that a
solution $(\bfU,T)$ of \req{nd:Bou}, \req{nd:bc} is \textit{regular on
  $[0,\tau]$} if
$(\bfU,\tht)\in L^\infty(0,\tau;\mathcal{V})\times L^\infty(0,\tau;V)$.
If $\tau=\infty$, we say that the solution is a \textit{global regular
  solution}.

\begin{Thm}[\cite{XWang05, Wang2007}]
\label{prop:3d:Bou}
Let $d=2,3$. Recalling the notation (\ref{eq:H-def})--(\ref{eq:W-def}) let
$(\bfU_0,\tht_0)\in \mathcal{H}\times H$ and $T_0:=\tht_0+(1-x_d)$.
\begin{enumerate}[(i)]
\item (Global Existence of Weak Solutions) For any $\tau>0$,
  there exists $(\bfU,T)$ such
  \begin{align*}
    &\bfU \in L^\infty(0,\tau;\mathcal{H}) \cap L^2(0,\tau;\mathcal{V}) 
      \cap C_w([0,\tau];\mathcal{H}),
      \qquad
      \od{\bfU}{t} \in L^{4/3}(0,\tau;\mathcal{V}'),
    \\
    &\theta \in L^\infty(0,\tau;{H}) \cap L^2(0,\tau;{V}) 
      \cap C_w([0,\tau];{H}),
      \qquad 
      \od{\theta}{t} \in L^{4/3}(0,\tau;{V}'), 
  \end{align*}
  where $\tht$ is related to $T$ by the relation
  \req{perturb:var}. This $(\bfU, T)$ satisfies
  \eqref{nd:Bou}--\eqref{nd:bc} in the usual weak sense, and maintains
  the initial condition $(\bfU_0,T_0)$. Moreover, the following energy
  inequality and maximum principle are satisfied for all $t\leq\tau$:
  \begin{gather}\label{wk:mp}
    \frac{1}{\prN}\|\bfU(t)\|^2+2\int_0^t\|\del\bfU(s)\|^2ds
    \leq\frac{1}{\prN}\|\bfU_0\|^2+2\raN\int_0^t\lb \tht(s),u_3(s)\rb\:ds,\notag\\
    \|(T-1)^+(t)\|^2+2\int_0^t\|\del(T-1)^+(s)\|^2\:ds
    \leq\|(T_0-1)^+\|^2,\\
    \|T^-(t)\|^2+2\int_0^t\|\del T^-(s)\|^2\:ds
    \leq\|(T_0)^-\|^2,\notag
  \end{gather}
  where $(T-1)^+, T^-$ are defined as in \req{pos-neg:part}.

\item (Eventual Regularity of Weak Solutions) There exists a
   constant $K_0 > 0$ such that if $\prN\raN^{-1}\geq K_0$,
  then there exists a time $\tau^*>0$ for which all suitable weak
  solutions corresponding to $(\bfU_0,\tht_0)\in\mathcal{H}\times H$ become regular solutions
  on $[\tau^*,\infty)$. In particular, when $d=3$, there exist constants $\kap_1,\kap_2,\kap_3>0$ and $\kap_0'',\kap_1'',\kap_2''>0$, depending on $\Om$, such that
	\begin{align}\label{eventual:reg}
		\begin{split}
		&\Sob{\bfU(t)}{H^1}\leq\kap_1\raN,\quad \Sob{\bfU(t)}{H^2}\leq \kap_2\raN^{5/2},\quad \Sob{\bdy_t\bfU(t)}{}\leq \kap_3\raN^{7/2}\\
		&\Sob{\tht(t)}{}\leq \kap_0'',\quad\qquad\quad\Sob{\tht(t)}{H^1}\leq\kap_1''\raN^{5/2},\quad\Sob{\tht(t)}{H^2}\leq\kap_2''\raN^8.
		\end{split}
	\end{align}
for all $t\geq\tau^*$.
\end{enumerate}
\end{Thm}

  For the remainder of the article, we will therefore assume that $(\bfU,T)$,
  $(\bfUT,\tilde{T})$ have evolved for a sufficiently long time, so that $(\bfU,T)$ is a regular solution to either 
  \req{nd:Bou}--\req{nd:bc} or 
  \eqref{system:data-infinite}. 
  Physically speaking, the set-up of our study assumes that reality is represented exactly by a solution to  \req{nd:Bou}--\req{nd:bc} or  \eqref{system:data-infinite} and that we have been observing the system after the point in time at which it has become globally regular. Thus, for the purposes of our analysis, we will henceforth make the following standing hypotheses for the remainder of the paper.

\subsection*{Standing Hypotheses.}

Let $d=2,3$ and $\raN\geq1$. Let $\gamma_0>0$ be the constants from Theorem \ref{cor:infinite:prandtl}. When $\Pr=\infty$, we assume:
\begin{enumerate}
\item[$(I1)$] $T_0$ a.e. periodic in $x_1,x_{d-1}$ and $T_0|_{x_d=0}=0$, $T_0|_{x_d=1}=1$ (in the sense of trace);
\item[$(I2)$] $T_0\in L^\infty(\Om)$;
\item[$(I3)$] $(\bfU,T)$ is the unique global solution to \req{system:data-infinite} corresponding to $T_0$ and guaranteed by Theorem \ref{cor:infinite:prandtl}.
In particular, $(\bfU,T)$ satisfies 
	\[
	    \raN^{-1}\Sob{\bfU(t)}{H^2}
        \leq \gamma_0,\quad \|T(t)\|\leq \gamma_0',
	\]
for all $t>0$.\qed
\end{enumerate}
On the other hand, let $K_0>0$, $\kap_1,\kap_2,\kap_3>0$, and $\kap_0'',\kap_1'',\kap_2''>0$ be the constants in Theorem \ref{prop:3d:Bou} $(ii)$. 
When $\Pr<\infty$, we assume:
\begin{enumerate}
\item[$(F1)$] $\Pr\raN^{-1}\geq K_0$;
\item[$(F2)$] $(\bfU,T)$ is the unique regular solution to \req{nd:Bou}, \req{nd:bc};
\item[$(F3)$] When $d=3$, $\bfU(t)$ satisfies
	\[
		\Sob{\bfU(t)}{H^1}\leq\kap_1\raN,\quad \Sob{\bfU(t)}{H^2}\leq \kap_2\raN^{5/2},\quad \Sob{\bdy_t\bfU(t)}{}\leq \kap_3\raN^{7/2},
	\]
	for all $t\geq0$.
\item[$(F4)$] When $d=3$, $T$ satisfies
	\[
		\Sob{T(t)}{}\leq\kap_0',\quad\Sob{T(t)}{H^1}\leq\kap_1'\raN^{5/2},\quad\Sob{T(t)}{H^2}\leq\kap_2'\raN^8,
	\]
	for all $t\geq0$, where $\kap_j'=\kap_j''+2|\Om|^{1/2}+1$, $j=0,1,2$.\qed
\end{enumerate}

\begin{Rmk}
Since we have non-dimensionalized our variables, we point out that although the bounds in $(F3), (F4)$ are derived for the case $d=3$, they are also valid for the case $d=2$, up to constants, provided that $\raN\geq1$, which have have assumed as one of our standing hypotheses.
\end{Rmk}
\section{Infinite-Prandtl Assimilation} 
\label{section:infinite-prandtl-assimilation}

We will first treat the Data Assimilation problem for the infinite
Prandtl number system \req{system:data-infinite},
\req{system:nudge-infinite}. Hence, throughout this section we assume that $\prN=\infty$, i.e., $(I1)-(I3)$ holds.

A rigorous mathematical analysis is
performed in dimensions $d=2,3$, while the numerical component of our
studies are carried out for $d=2$. Due to the structure of the nudged system \req{system:nudge-infinite}, we do not have a maximum principle. Instead, we require only that \req{system:nudge-infinite} have a well-defined solution in the weak sense, which one can do by establishing that the differences $\bfW=\bfUT-\bfU$, $S=\T-T$ satisfy their respective evolution in the weak sense. Since one of the relevant a priori estimates to this end are performed below for the proof of Theorem \ref{thm:infty}, we simply state this result as the following theorem.

\begin{Thm}\label{thm:nudge:infinite}
Let $\mu>0$ and $\{\lambda_n\}_{n=1}^\infty$ be as in \req{eq:eigenpairs}. Let $\T_0\in L^2(\Om)$ a.e. in $\Om$ such that $\til{T}_0$ is
  a.e. $L$-periodic in $x_1,x_{d-1}$ with
  ${\T}_0|_{x_d=0}=0,{\T}_0|_{x_d=1}=1$ (in the sense of trace). Suppose $N>0$ satisfies
$\frac{1}{4}\lambda_N \ge \mu$.
Then there exists a unique $(\bfUT,\til{T})$ satisfying \req{system:nudge-infinite} in the weak sense
  such that
  \begin{align*}
    \bfUT \in L^\infty(0,\tau;\mathcal{W}),
    \quad 
    \T\in L^\infty(0,\tau;L^2(\Om))\cap 
         L^2(0,\tau;H^1(\Om))\cap C_w([0,\tau],L^2(\Om)),
  \end{align*}
  for all $\tau>0$. 
\end{Thm}

\subsection{Synchronization} 
\label{sec:analysis}

We are now ready to state and prove the main theorem of this section. We show that
\eqref{system:data-infinite} and \eqref{system:nudge-infinite}
synchronize under certain conditions detailed below.

\begin{Thm}[Infinite-Prandtl Synchronization]
\label{thm:infty}
Let $\mu>0$ and $\{\lam_n\}_{n=1}^\infty$ be as in \req{eq:eigenpairs}. Let $N>0$ satisfy $ \frac{1}{4}\lambda_N \geq\mu$ and $\til{T}_0$ be given as in Theorem \ref{thm:nudge:infinite}. Let $(\bfUT,\T)$ be the corresponding unique solution to \req{system:nudge-infinite} guaranteed by Theorem \ref{thm:nudge:infinite}. There exists a constant $C_0=C_0(\Om,\|T_0\|)> 0$ such that if
\begin{align}
\label{eq:mu-requirement}
   \mu \ge {C_0}\raN^2,
\end{align}
then for all $t\geq0$
\begin{align}\label{synch:infinite}
\|(\T - T)(t)\|^2 + \raN^{-2}\|(\bfUT - \bfU)(t)\|^2_{H^2} \leq C_1e^{-\mu t},
\end{align}
for
$C_1=\|\T_0-{T}_0 \|$.  
\begin{proof} 
  Let $S = \T - T$, $\bfW = \bfUT - \bfU$, and $q=\tilde{p}-p$.
  Subtracting \eqref{system:data-infinite} from
  \eqref{system:nudge-infinite} yields the system
  \begin{align}\label{eq:diff:DA}
    \begin{gathered} 
    -\Delta \bfW = -\nabla q + \raN \bfe_3 S,
      \quad \nabla\cdot\bfW = 0,\\
    \partial_t S + \bfUT \cdot \nabla S + \bfW \cdot \nabla T - \Delta S = -\mu P_N S,\\
    \bfW|_{x_3 = 0} = \bfW|_{x_3 = 1} = \mathbf{0},
      \quad S|_{x_3 = 0} = S|_{x_3 = 1} = 0,
      \quad \bfW, S\textrm{ are periodic in }x_1 \textrm{ and } x_2,
  \end{gathered}
\end{align}
with the initial condition $S(\bfX,0) = \T_0(\bfX) - T_0(\bfX)$.  The
momentum equation in \eqref{eq:diff:DA} satisfies Lemma
\ref{lem:stokes}, so there exists a constant $C > 0$ such that
\begin{align}
\|\bfW\|_{H^2} + \|q\|_{H^1} &\leq C\Sob{\raN\mathbf{e}_3 S}{}.
\label{eq:stokes-lemma-applied}
\end{align}
Therefore, to establish \req{synch:infinite}, it is sufficient to show
that $S\rightarrow 0$ with an exponential rate in $L^2(\Omega)$.

Upon multiplying the $S$ equation in \eqref{eq:diff:DA} by $S$ and
integrating over $\Omega$, we obtain
\begin{align}\label{energy:balance:infty}
  \frac{1}{2} \frac{d}{dt} \|S\|^2 + \|\del S\|^2  + \mu \|P_N(S)\|^2
  &= -\int_\Omega (\bfW \cdot \del T)S\:d\bfX.
\end{align}
Assume that $\mu$ is chosen sufficiently large so that
\eqref{eq:mu-requirement} holds,
where $C_0>0$ is, as of yet, unspecified.  After integrating by parts,
an application of the Sobolev embedding
$H^2(\Omega)\imb L^\infty(\Omega)$, \req{eq:stokes-lemma-applied}, and then 
$(I3)$ of the Standing Hypotheses imply 
\begin{align}\label{2d:nlt:est}
\left| \int_\Omega (\bfW \cdot \del T) S\:d\bfX\right|
&\leq \|\bfW\|_{L^\infty}\|T\|\|\del S\|
\leq C \raN\|S\|\|T\|\|\del S\|
\notag\\
& \leq \frac{1}{2}\|\del S\|^2 +  C \raN^2\|T\|^2\|S\|^2
\leq \frac{1}{2}\|\del S\|^2+C\raN^2\gamma_0'\|S\|^2.
\end{align}
On the other hand, due to \eqref{eq:pn+qn=f} and the inverse
Poincar\'e inequality, and since $N>0$ satisfies $\frac{1}4\lam_N\geq\mu$, it follows that
\begin{align}
\label{2d:poincare}
   \frac{1}{2}\|\del S\|^2 + \mu\|P_NS\|^2
&= \frac{1}{2}\|\del S\|^2 - \mu\|Q_NS\|^2 + \mu \|S\|^2
\notag\\
&\geq\left(\frac{1}2-\frac{\mu}{\lambda_N}\right)\|\del S\|^2+\mu\|S\|^2
\geq \mu\|S\|^2.
\end{align}
Joining \eqref{energy:balance:infty} with \eqref{2d:nlt:est} and
\eqref{2d:poincare} and combining like terms, we arrive at
\begin{align*}
  \frac{d}{dt}\|S\|^2 + 2\left(\mu  - C\raN^2|\Om|\right)\|S\|^2  \leq 0.
\end{align*}
Finally, by Gronwall's inequality and the second
condition in \eqref{eq:mu-requirement}, we deduce that
\begin{align}
\|S(t)\|^2 &\leq \|S(0)\|^2\exp\left(-2\mu t+C\raN^2\gamma_0't\right)\leq \|S(0)\|^2\exp\left(-\mu t\right),
\label{eq:DA:est:1}
\end{align}
which completes the proof upon choosing $C_0=C\gamma_0'$.
\end{proof}
\end{Thm}

Theorem \ref{thm:infty} shows that given $\raN$ and reasonable
boundary conditions $T_0$ and $\T_0$, we can always choose $\mu$ and
$N$ large enough so that $(\bfU,T)$ and $(\bfUT,\T)$ eventually match
in the infinite-time limit.  In the next subsection, we use numerical
simulations to verify that this is indeed the case.

\subsection{Numerical Results at Infinite Prandtl number} 
\label{section:numerical-results-infinite}

Rather than focusing on a handful of highly turbulent
three-dimensional simulations, for the same computational cost we
consider more detailed and extensive simulations in the
two-dimensional setting where $\Omega = [0,L]\times[0,1]$ with
coordinates $\bfX = (x_1,x_3)$ in order to search through the relevant
parameters.  We simulate (\ref{system:data-infinite}) and
(\ref{system:nudge-infinite}) using a stream function
formulation with \emph{Dedalus} \cite{2016ascl.soft03015B}, a Python
package that uses pseudospectral methods to solve partial differential
equations on spectrally representable domains. All of the simulations in this section and all that follow are completed with a 4-stage 3rd order Runge-Kutta implicit-explicit time stepping scheme that treats the linear terms implicitly and the nonlinear terms explicitly.  Each simulation runs with $L=4$ and $256$ Fourier grid points in the horizontal and $128$ Chebyshev points in the vertical with a standard $3/2$ dealiasing.  All simulations were run until the time-averaged Nusselt number and other pertinent statistics were temporally well-converged for all cases considered here.

Choosing the initial conditions $T_0$ and $\T_0$ in our numerical
experiments is nontrivial.  For $T_0$, we have two options.
\begin{itemize}
\item Set $T_0(x_1,x_3) = 1 - x_3 + \varepsilon(x_1,x_3)$ for some
  small perturbation function $\varepsilon:\Psi\rightarrow\RR$.  This
  begins the simulation close to the conductive state $1 - x_3$,
  which---though not physically relevant---is a suitable starting
  point for initial experiments.
\item Load $T_0$ from the final state of a previous simulation with
  similar parameters.  If this previous simulation has run long enough, $T$
  will thus begin in a reasonable state for the new set of parameters.
\end{itemize}
Finding a suitable choice for $\T_0$ is even less obvious.  Setting
$\T_0(x_1,x_3) = 1 - x_3$ assumes no prior knowledge of $T$, and results in an overly stiff initialization due to the strength
of the initial temperature difference, and is therefore numerically
impractical.  In addition, initiating $T$ at this conductive state
ignores observations taken at the initial time $t=0$ which would not
be advantageous in practice.  On the other hand, setting
$\T_0 = P_N(T_0)$ for large $N$ results in a very weak temperature
difference so that (\ref{system:data-infinite}) and
(\ref{system:nudge-infinite}) are nearly synchronized at $t = 0$.  As
a compromise, we set $\T_0$ as a low-mode projection of $T_0$ such as
$\T_0 = P_4(T_0)$.  This permits some initial knowledge of the true
state without driving it directly to the `truth' immediately.

Values of the Rayleigh number $\raN$ that are of interest for mantle
convection are typically between $10^7$ and $10^8$
\cite{turcotteschubert2014,Blankenbachetal1989}.  However, the greater
value for $\raN$, the more computationally expensive the simulation,
and it is numerically unstable to start $T$ in a state such as
$T_0(x_1,x_3)\approx 1-x_3$ at large $\raN$.  We therefore select
logarithmically spaced $\raN$ values from $10^4$ to $10^9$ and run a
simulation for each $\raN$, one after another.  For the very first
simulation ($\raN = 10^4$), we set
$T_0(x_1,x_3) = 1 - x_3 + \varepsilon(x_1,x_3)$, as discussed above,
and evolve the systems forward until reaching a nearly steady state.
Thereafter, we set $T_0$ as the final $T$ from the previous simulation
and increment increase $\raN$.  This yields realistic initial conditions for
any given $\raN$ in the range considered. Figure
\ref{fig:basic-before-after} shows an example of $T_0$ and $\T_0$,
together with subsequent end states of $T$ and $\T$ for
$\raN \approx 5.22\times 10^7$.
\begin{figure}
\centering
\begin{subfigure}{.45\textwidth}
    \centering
    \includegraphics[width=\textwidth]{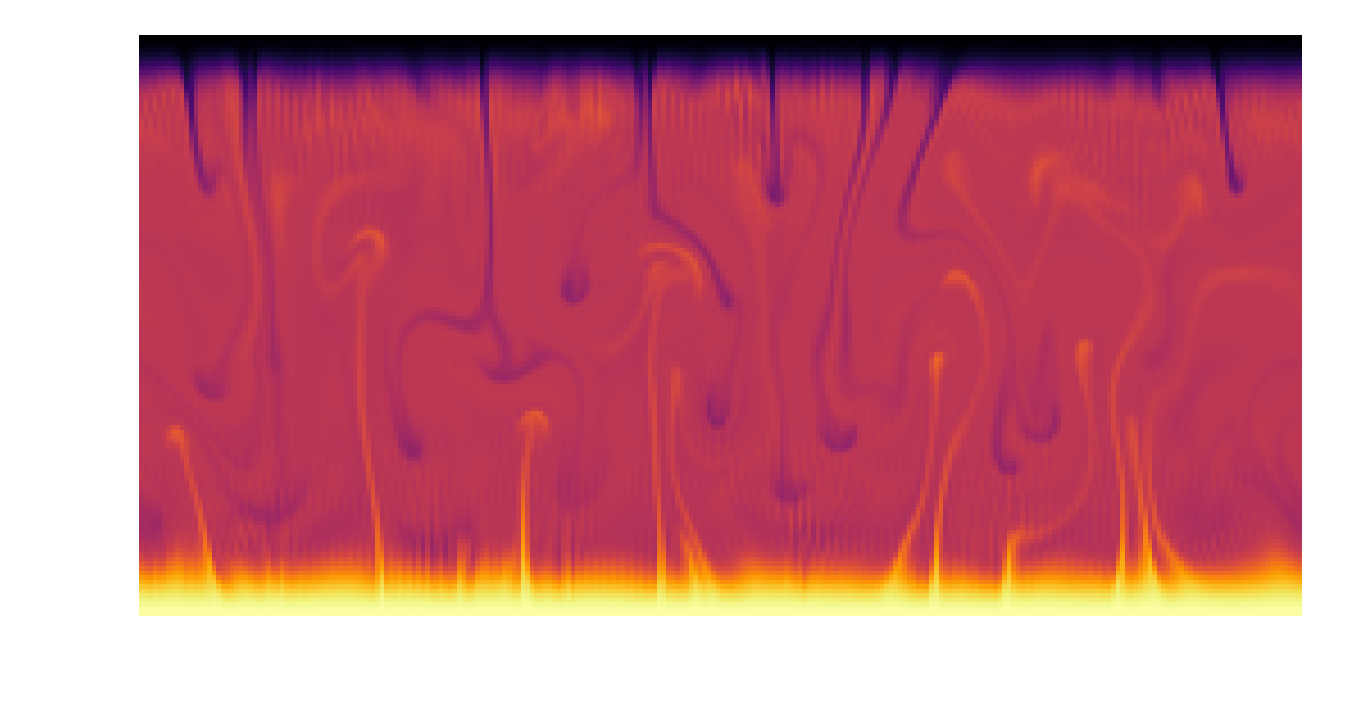}
\end{subfigure}
$\quad\longrightarrow$
\begin{subfigure}{.45\textwidth}
    \centering
    \includegraphics[width=\textwidth]{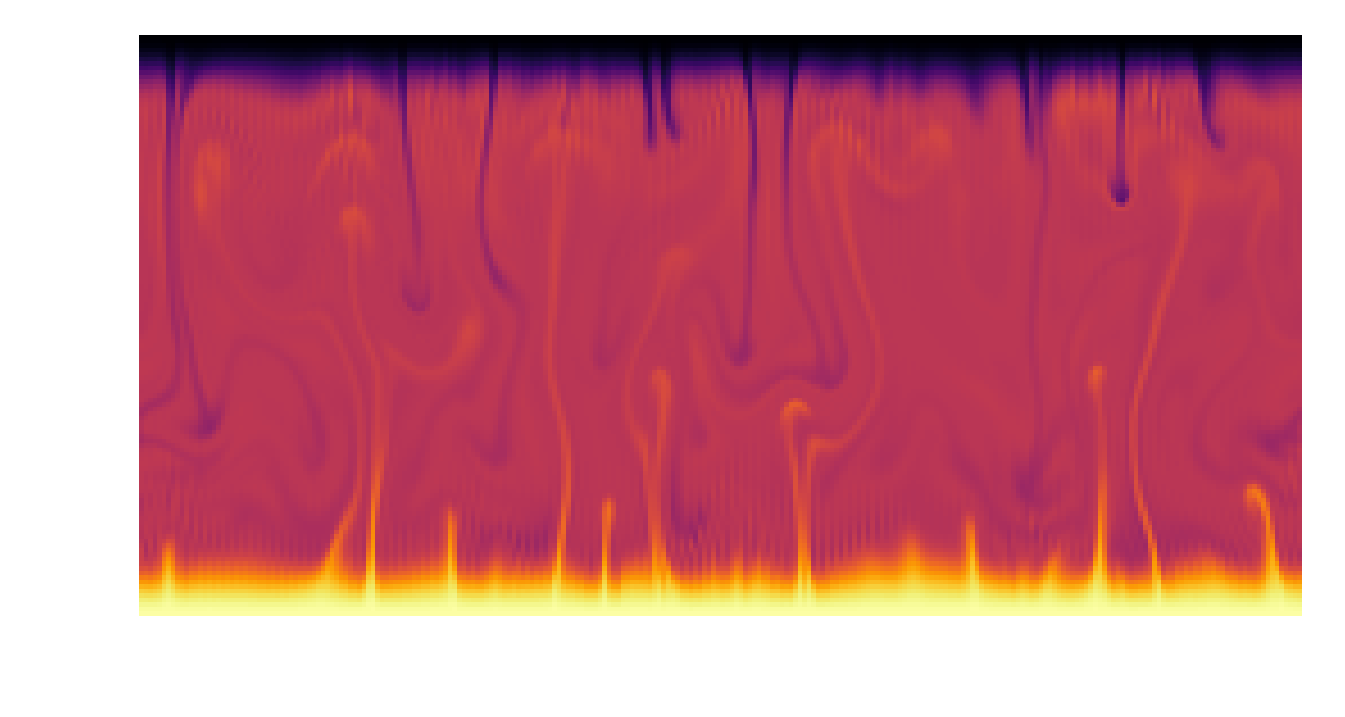}
\end{subfigure}
\\
\begin{subfigure}{.45\textwidth}
    \centering
    \includegraphics[width=\textwidth]{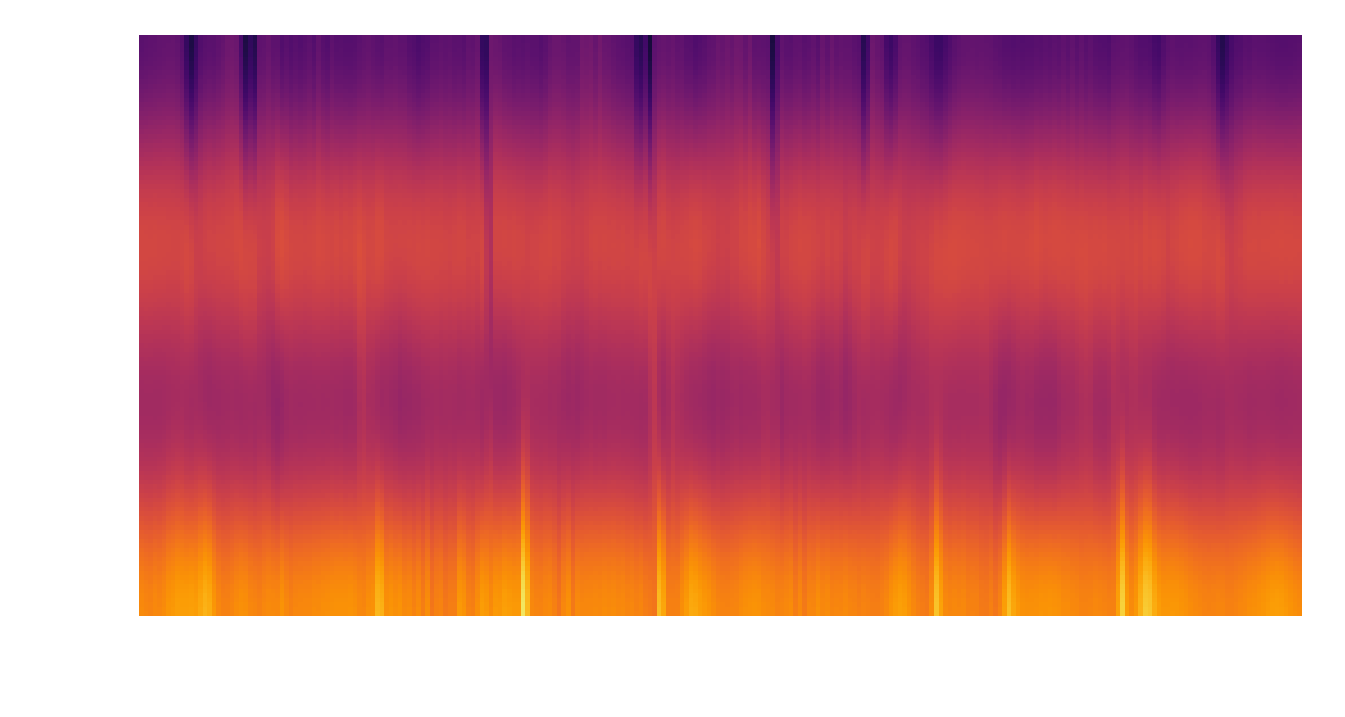}
\end{subfigure}
$\quad\longrightarrow$
\begin{subfigure}{.45\textwidth}
    \centering
    \includegraphics[width=\textwidth]{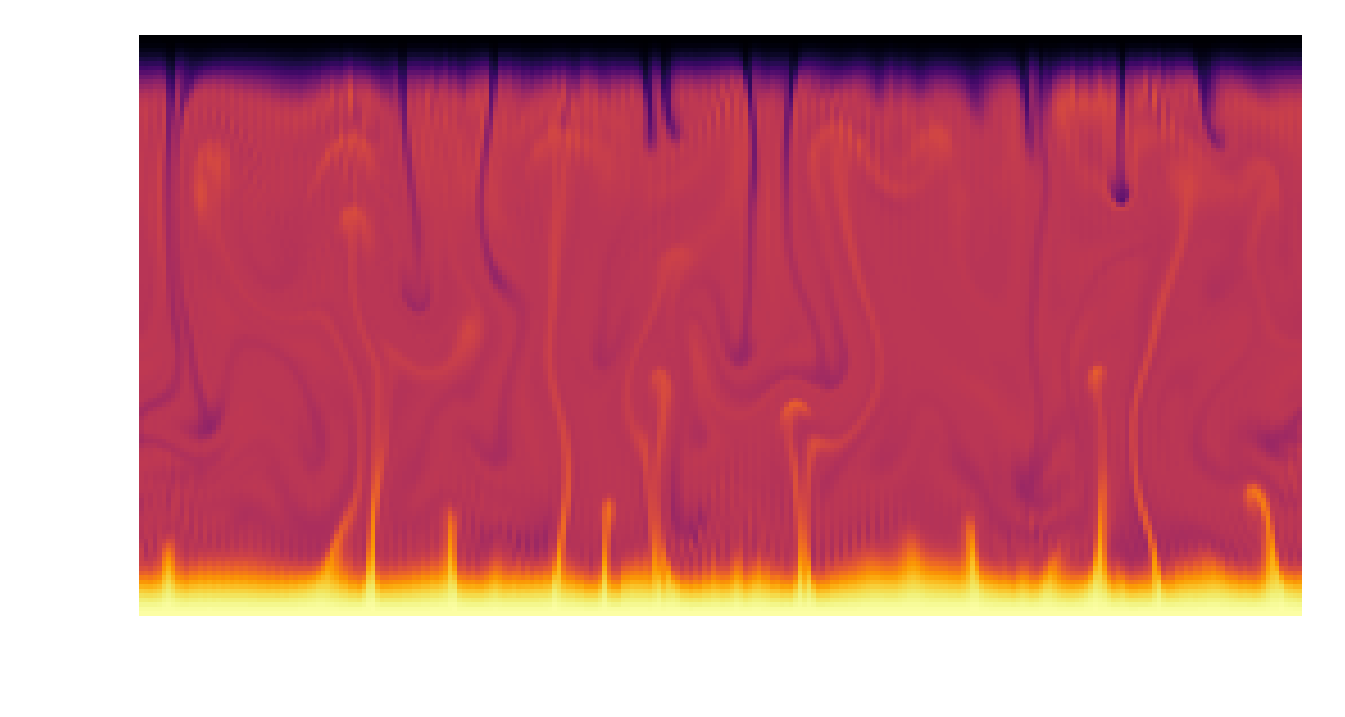}
\end{subfigure}
\caption[Typical ``before'' and ``after'' snapshots of a successfully
assimilated setup.]{Typical ``before'' and ``after'' snapshots of a
  successfully assimilated setup.  On the top, the initial condition
  $T_0$ compared with $T$ at the end of a simulation.  On the bottom,
  the initial projected temperature $\T_0 = P_4(T_0)$ and the eventual
  $\T$ at the end of the simulation.  Note that the post-simulation
  $T$ and $\T$ appear identical to the naked eye.  This particular
  simulation uses $\raN \approx 5.22\times 10^7$, $\mu=12,700$, and
  $N=32$.}
\label{fig:basic-before-after}
\end{figure}

To measure the synchronization of (\ref{system:data-infinite}) and
(\ref{system:nudge-infinite}), we keep track of the $L^2(\Omega)$,
$H^1(\Omega)$, and $H^2(\Omega)$ norms of $\T - T$ and $\bfUT - \bfU$
throughout the simulation.  Each error is normalized by dividing by
the norms of the truth system's variables.  For example, to measure
the temperature difference in $L^2(\Omega)$, we compute
$\frac{\|(\T - T)(t)\|}{\|T(t)\|}$ at each simulation time $t$.  For
simplicity, we write these errors without the denominators in all that
follows.

Figure \ref{fig:basic-convergence} shows how the tracked error norms
decrease in the simulation from Figure \ref{fig:basic-before-after}.
Though the analysis indicates that these norms should converge to
$\varepsilon_{\text{machine}}\approx 10^{-16}$, seeing the errors
decrease to about $10^{-9}$ is numerically satisfactory in this
setting, given the format in which the norms are calculated, i.e. some
errors do accumulate in the calculation of the norm itself.
\begin{figure}
\centering
\includegraphics[width=\textwidth]{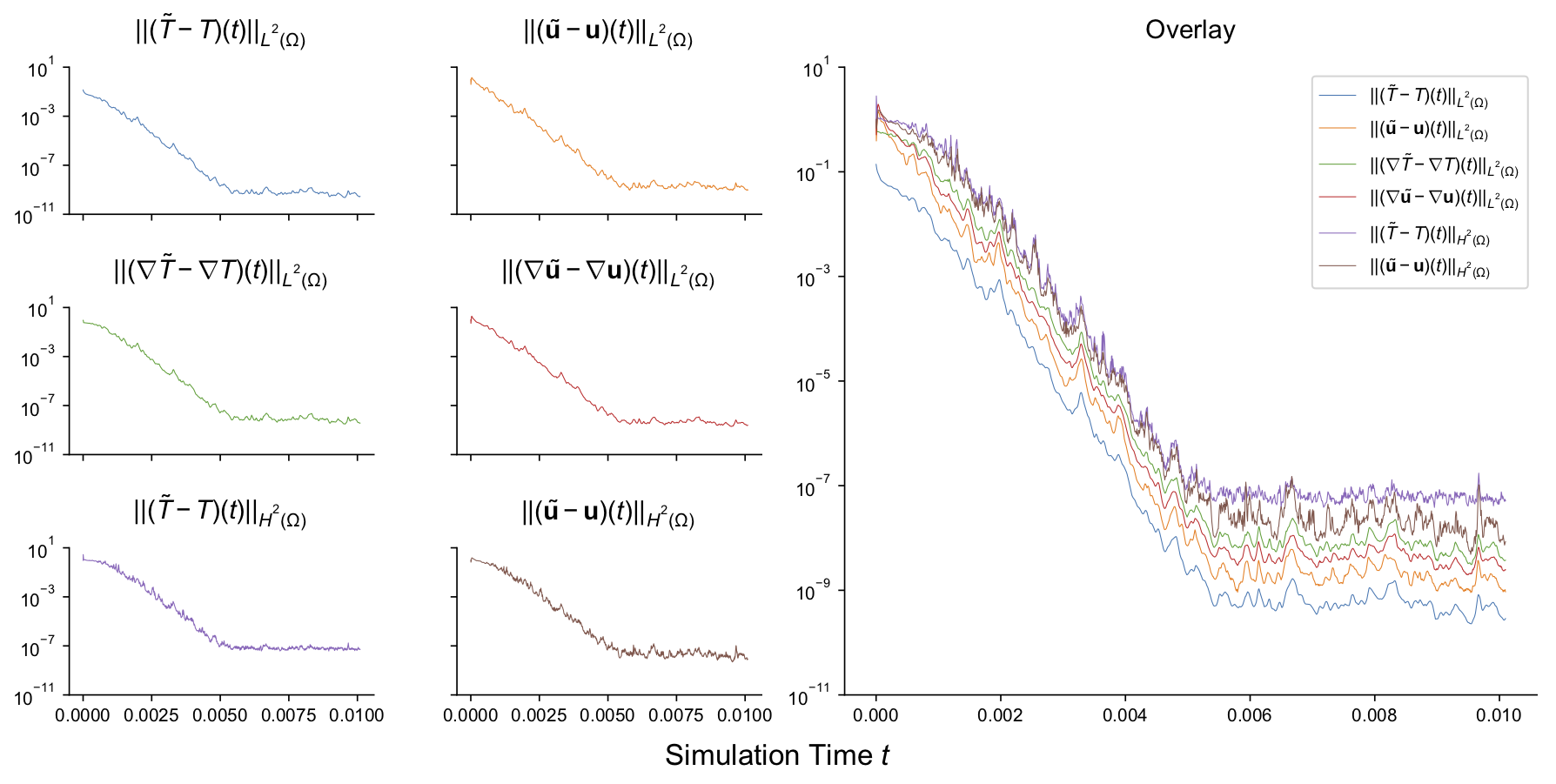}
\caption[Synchronization in various norms for a typical numerical
experiment with $\prN = \infty$.]{Synchronization in various norms for
  $\raN \approx 5.22\times 10^7$, $\mu=12,700$, and $N=32$ (the same
  conditions used in Figure \ref{fig:basic-before-after}).  The
  temperature and velocity differences decrease exponentially until
  flattening out at sufficiently small values.  The temperature
  difference is the smallest, which seems to be a consequence of using
  temperature-only observations for the assimilation.}
\label{fig:basic-convergence}
\end{figure}

\subsubsection{Dependencies of nudging parameter} 

Previous studies, \cite{AltafTitiGebraelKnioZhaoMcCabeHoteit15,
  FarhatJohnstonJollyTiti18}, of data assimilation for
Rayleigh-B\'enard convection fixed the value of the nudging parameter
at $\mu = 1$ in their experiments and instead focused on the effects of
other parameters, such as the number of projected modes.  However,
Theorem \ref{thm:infty} requires that $\mu$ be proportional to
$\raN^2$, so we do not expect synchronization with $\mu = 1$ as $\raN$
increases.  Indeed, Figure \ref{fig:bad-mu} shows that for $\raN$ in
the range of $4\times 10^7$ to $2\times 10^8$, using $\mu = 1$ results
in either extremely slow convergence, or in no convergence at all.

\begin{figure}
\centering
\includegraphics[width=.7\textwidth]{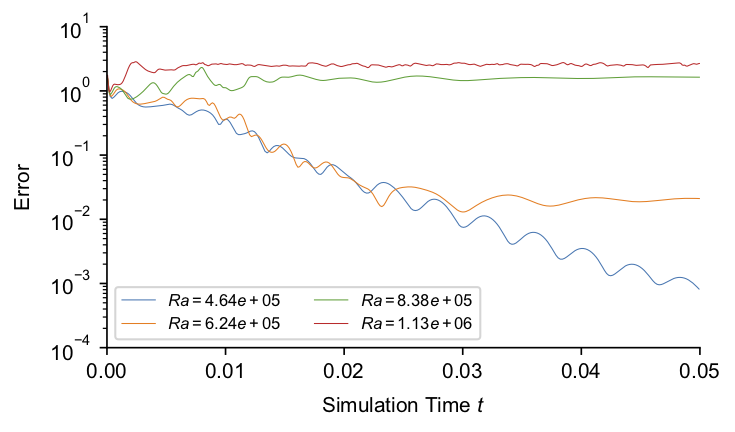}
\caption[Convergence (or lack thereof) with $\mu = 1$ and various
$\raN$.]{Convergence (or lack thereof) of
  $\|(\bfUT-\bfU)(t)\|_{H^2} + \|(\T-T)(t)\|_{H^2}$
  with $\mu = 1$ and $N = 32$ for various values of $\raN$.  For lower
  $\raN$, setting $\mu = 1$ is sufficient to slowly induce
  synchronization; however, as $\raN$ increases, $\mu = 1$ quickly
  becomes too weak to nudge the assimilating system toward the truth.}
\label{fig:bad-mu}
\end{figure}

To explore the relationship between $\raN$, $\mu$, and $N$ needed for
synchronization to occur, we fix two of the parameters at a time and
run several simulations with various values for the remaining
parameter.  First, we fix $\raN$, set $N = 32$, and vary $\mu$ until
synchronization can be observed within $0.005$ units of simulation
time (a simulation of this length usually requires several thousand
iterations).  Table \ref{table:mu_values} records the smallest value
of $\mu$ where synchronization was observed; Figure \ref{fig:vary-mu}
shows the convergence rates for a few different $\mu$ values when
$\raN \approx 3.89\times 10^7$ and $\raN\approx 7.02\times 10^7$.

\begin{figure}
\centering
\begin{subfigure}{.49\textwidth}
    \centering
    \includegraphics[width=\textwidth]{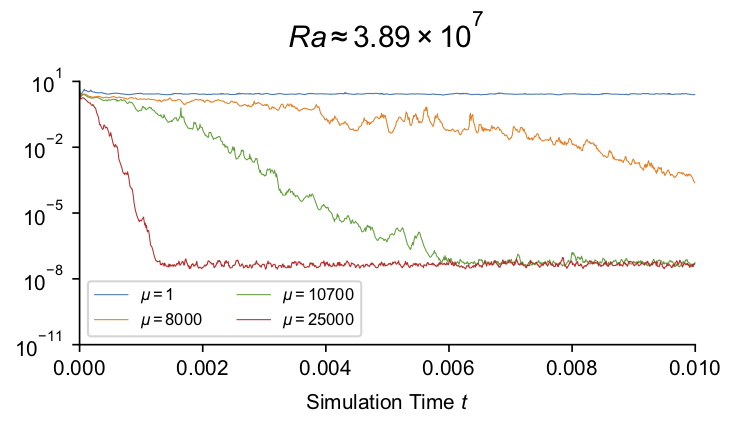}
\end{subfigure}
\begin{subfigure}{.49\textwidth}
    \centering
    \includegraphics[width=\textwidth]{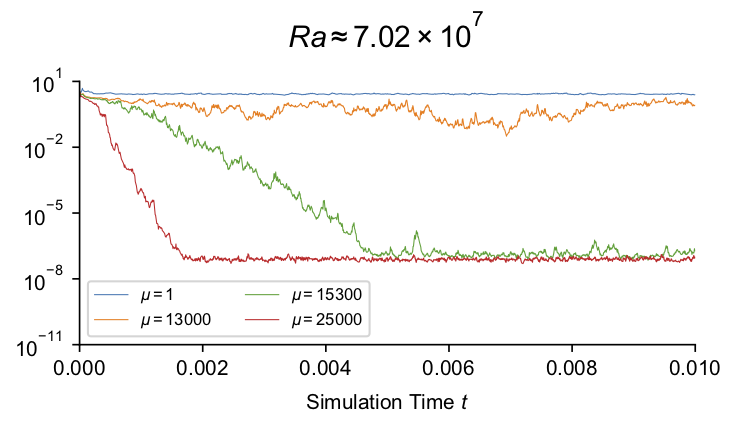}
\end{subfigure}
\caption[Synchronization with $\prN = \infty$, $N = 32$, and various
$\mu$ for two different values of
$\raN$.]{$\|(\bfUT-\bfU)(t)\|_{H^2} +
  \|(\T-T)(t)\|_{H^2}$
  with $N = 32$ and various $\mu$ for two different values of $\raN$.
  We begin to see satisfactory convergence around $\mu = 10,700$ and
  $\mu=15,300$ for $\raN \approx 3.89\times 10^7$ and
  $\raN \approx 7.02\times 10^7$, respectively.}
\label{fig:vary-mu}
\end{figure}

\begin{table}
\centering
\begin{tabular}{r||r}
\multicolumn{1}{c||}{$\raN$} & \multicolumn{1}{c}{$\mu$} \\ \hline
$1.1937766 \times 10^7$ &  $6,300$ \\
$1.6037187 \times 10^7$ &  $6,500$ \\
$2.1544346 \times 10^7$ &  $7,900$ \\
$2.8942661 \times 10^7$ &  $9,400$ \\
$3.8881551 \times 10^7$ & $10,700$ \\
$5.2233450 \times 10^7$ & $12,700$ \\
$7.0170382 \times 10^7$ & $15,300$ \\
$9.4266845 \times 10^7$ & $19,400$ \\
$1.26638017\times 10^8$ & $24,900$ \\
$1.70125427\times 10^8$ & $29,900$
\end{tabular}
\caption[Minimal values of $\mu$ that result in synchronization within about 
0.005 units of simulation time for various $\raN$ with $N = 32$.]{Minimal values 
  of $\mu$ that result in synchronization within about 0.005 units of simulation 
  time for the given $\raN$ with $N = 32$.  These values represent the edge of what 
  works when $N = 32$: a lower $\mu$ may not stimulate convergence, but a larger 
  $\mu$ will.  See Figure \ref{fig:mu_fit-N_32} for a quadratic least-squares fit 
  of this data.}
\label{table:mu_values}
\end{table}

Given the requirement \eqref{eq:mu-requirement} from Theorem
\ref{thm:infty}, it is surprising to discover that---based on Table
\ref{table:mu_values}---the relationship between $\mu$ and $\raN$ is
more linear than quadratic.  Indeed, a least squares fit of the data
to a general quadratic of the form $f(x) = a + bx + cx^2$ yields
coefficients of $a \approx 4.03\times 10^3$,
$b\approx 1.77\times 10^{-4}$, and $c\approx -1.37\times 10^{-13}$
(essentially $c = 0$).  See Figure \ref{fig:mu_fit-N_32}.

\begin{figure}[H]
\centering
\includegraphics[width=.7\textwidth]{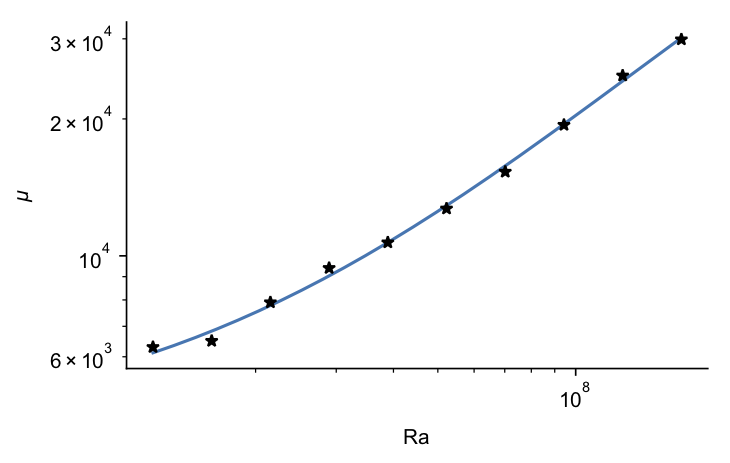}
\caption{Quadratic (but nearly linear) least-squares fit for the data in Table \ref{table:mu_values}.}
\label{fig:mu_fit-N_32}
\end{figure}

\subsubsection{Relation between relaxation and number of observables} 

The relaxation parameter $\mu$ is a system parameter without a clear
physical interpretation.  The number of projected Fourier modes $N$,
on the other hand, indicates the amount of data that is ``visible'' to
the assimilating system.  Therefore, a lower bound on $N$ represents
how much data is required in order to maintain an accurate model.
Fixing $\mu$ as given by Table \ref{table:mu_values}, we decrease $N$
to see how it affects synchronization.  See Figure \ref{fig:vary-N}.
Note that as expected, the less modes retained in the projection, the
slower synchronization is achieved, and if a sufficiently low number
of modes ($N \leq 6$) is observed then synchronization never
occurs.  In Table~\ref{table:N_values} we summarize our findings
concerning the number of modes $N$ needed as a function of $\raN$
(given a suitably large choice for $\mu$).

\begin{figure}
\centering
\begin{subfigure}{.49\textwidth}
    \centering
    \includegraphics[width=\textwidth]{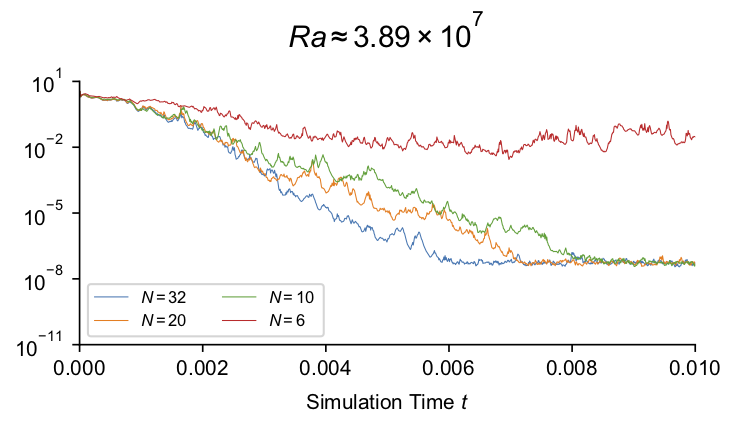}
\end{subfigure}
\begin{subfigure}{.49\textwidth}
    \centering
    \includegraphics[width=\textwidth]{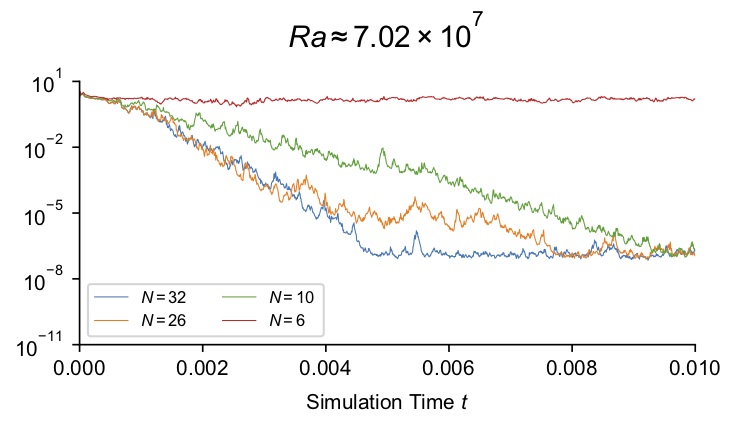}
\end{subfigure}
\caption[Synchronization for two different values of $\raN$, with $\mu$ as listed 
in Tables \ref{table:mu_values} and \ref{table:N_values} and for various values 
of $N$.]{$\|(\bfUT-\bfU)(t)\|_{H^2} + \|(\T-T)(t)\|_{H^2}$ for two 
  different values of $\raN$, with $\mu$ as listed in Tables \ref{table:mu_values} 
  and \ref{table:N_values} and for various values of $N$. For both 
  $\raN \approx 3.89\times 10^7$ and $\raN \approx 7.02\times 10^7$, $N=10$ 
  appears to be the least number of modes that results in synchronization.}
\label{fig:vary-N}
\end{figure}

\begin{table}[H]
\centering
\begin{tabular}{r|r||c}
\multicolumn{1}{c|}{$\raN$} & \multicolumn{1}{c||}{$\mu$} & $N$ \\ \hline
$1.1937766 \times 10^7$ &  $6,300$ & $6$ \\
$1.6037187 \times 10^7$ &  $6,500$ & $8$ \\
$2.1544346 \times 10^7$ &  $7,900$ & $8$ \\
$2.8942661 \times 10^7$ &  $9,400$ & $8$ \\
$3.8881551 \times 10^7$ & $10,700$ & $10$ \\
$5.2233450 \times 10^7$ & $12,700$ & $10$ \\
$7.0170382 \times 10^7$ & $15,300$ & $10$ \\
$9.4266845 \times 10^7$ & $19,400$ & $10$ \\
$1.26638017\times 10^8$ & $24,900$ & $12$ \\
$1.70125427\times 10^8$ & $29,900$ & $14$
\end{tabular}
\caption[Minimal values of $N$ that result in synchronization within 
about 0.005 units of simulation time for the $\raN$ and $\mu$ 
in Table \ref{table:mu_values}.]{Minimal values of $N$ that result 
  in synchronization within about 0.01 units of simulation time for 
  the given $\raN$ and $\mu$.  In general, as long as an appropriate $\mu$ 
  is chosen, $N = 16$ is sufficient for physically relevant values of $\raN$.}
\label{table:N_values}
\end{table}

\subsubsection{Summary} 

The numerical simulations verify that over physically relevant values
of $\raN$, we can pick $\mu$ large enough to nudge
\eqref{system:data-infinite} toward \eqref{system:nudge-infinite}; in
particular, the experiments confirm Theorem \ref{thm:infty} and show
that the conditions stated therein are stronger than necessary.
Furthermore, synchronization can be obtained with a relatively low
number of modes even for large $\raN$, as long as $\mu$ is chosen
large enough.  We took the approach of varying $\mu$ first and then
$N$, but it is just as feasible to vary $\mu$ for a fixed $\raN$ and
$N$.

\section{Finite-Prandtl Assimilation} 
\label{section:finite-prandtl-assimilation}
We now turn to the case when $\prN$ is finite, but large. Thus, throughout this section, we will assume that $\prN<\infty$, so that we automatically have that $(F1)-(F4)$ holds.

 As in
Section \ref{section:infinite-prandtl-assimilation}, we first
state the relevant well-posedness result for the nudged equation. We then rigorously establish synchronization of the nudged signal with the
true signal, then proceed with a numerical study for the two-dimensional setting.

It is not immediately clear how to properly adapt the argument for
\eqref{system:data-infinite}--\eqref{system:nudge-infinite} in
\cite{Temam1997, XWang05} to establish a maximum principle (see Lemma
\ref{lem:mp}) for the assimilating variable $\T$ due to the presence
of the additional term $-\mu P_N(\T-T)$ in
\eqref{system:nudge-infinite} (cf. \cite{JollyMartinezSadigovTiti18}
for a case where it is necessary to develop a maximum principle in
spite of this term).  However, since we are assuming $(\bfU,T)$ is a
regular solution to \eqref{nd:Bou}--\eqref{nd:bc}, it is
straightforward to verify global existence and uniqueness for the
associated nudged system by instead considering the corresponding
system for the difference $(\bfW, S)$, where $\bfW=\bfUT-\bfU$ and
$S=\T-T$.  In this setting, we need only appeal to a maximum principle
for $T$, rather than for $\T$.  The well-posedness of the nudged
system then follows in a standard fashion, so that we opt to only
state the result here and refer the reader to \cite{FarhatJollyTiti15}
for the appropriate details.

\begin{Thm}\label{thm:finite-nudge}
Let $\mu>0$ and $\{\lam_n\}_{n=1}^\infty$ be as in \req{eq:eigenpairs}.  Let $(\bfUT_0,\T_0-(1-x_d))\in\mathcal{V}\times V$. Suppose $N>0$ satisfies $\frac{1}4\lam_N\geq\mu$. Then there exists a unique solution $(\bfUT,\T)$ satisfying \req{eq:B:DA} such that
	\[
		(\bfUT,\til{T}-(1-x_d))\in L^\infty(0,\tau;\mathcal{V})\times L^\infty(0,\tau;V),
	\]
for all $\tau\geq0$.
\end{Thm}

\subsection{Synchronization} 

We will establish the finite Prandtl analog to Theorem \ref{thm:infty}. For this, it will be convenient to establish the following stability estimate first.

\begin{Lem}\label{lem:stab:3d}

Assuming the conditions of Theorem \ref{thm:finite-nudge}, there exists a constant $C_0=C_0(\Om)>0$ such that if
    \begin{align}\label{3d:lam:hyp}
	 \mu\geq \left( \frac{1}{2} + \raN^2\right)\prN,
    \end{align}
then
\begin{align} 
  &\|(\bfUT-\bfU)(t)\|^2+\|(\T-T)(t)\|^2\notag\\
      &\leq  \exp \left( - t\prN +C_0 \int_0^t
        \left( \frac{1}{\prN^3}\|\del\bfU(s)\|^4
          + \frac{1}{\prN}\|T(s)\|_{L^6}^4\right) ds\right) 
          ( \|\bfUT_0-\bfU_0\|^2+ \|\T_0-T_0\|^2),
        \label{eq:f:p:syncon:ineq}
\end{align}
holds for $t\geq0$.
\begin{proof}
  Similar to the analysis performed for Theorem \ref{thm:infty}, we
  consider the solution $(\bfW,S)$ to the difference system
\begin{align}
  \label{eq:diff:DA:2}
  \begin{gathered}
      \frac{1}{\prN}\left[\partial_t\bfW + (\bfUT\cdot\del)\bfW 
        + (\bfW\cdot\del)\bfU\right] - \Delta\bfW 
      = -\nabla q + \raN\mathbf{e}_3 S,\quad  \nabla\cdot\bfW = 0,\\
    \partial_t S + \bfUT \cdot \nabla S + \bfW \cdot \nabla T - \Delta S 
    = -\mu P_N S,\\
    \bfW|_{x_3 = 0} = \bfW|_{x_3 = 1} = \mathbf{0},
      \quad S|_{x_3 = 0} = S|_{x_3 = 1} = 0,
      \quad \bfW, S\textrm{ are $L$-periodic in }x_1 \textrm{ and } x_2.
  \end{gathered}
\end{align}
Then, multiplying the second equation by $S$ integrating and arguing as
in \eqref{2d:poincare} we obtain
\begin{align}
   \ddt \|S\|^2+ 2 \|\del S\|^2  + 2 \mu\|S\|^2 
  &\leq2\left|\int \bfW \cdot \nabla T Sdx\right|
    =:I.
   \label{eq:DA:est:5}
\end{align}
On the other hand, an analogous energy calculuation for $\bfW$ yields
\begin{align}
  \ddt \|\bfW\|^2 &+ 2\prN\|\del\bfW\|^2 + 2\prN\|\bfW\|^2  \notag\\
         &\leq 2 \left|\int \bfW \cdot \nabla \bfU \cdot \bfW dx \right| 
           + 2\raN \prN \left| \int  {\bfe_3} S \cdot \bfW dx \right|
           =:II+III.
     \label{eq:DA:est:6}
\end{align}

We estimate the right hand side of \eqref{eq:DA:est:5} using
H\"older's inequality, the Gagliardo-Nirenberg interpolation inequality, and Young's
inequality to obtain
\begin{align*} 
  I
     \leq& 2\| \bfW\|_{L^3}\|\del S\|\|T\|_{L^6} 
     \leq C  \|\bfW\|^{1/2} \|\del\bfW\|^{1/2}\|\del S\|\|T\|_{L^6} \notag\\
     \leq &C  \|T\|_{L^6}^2\|\bfW\|\|\del\bfW\|+ \|\del S\|^2 
    \leq \frac{C}{\prN} \|T\|_{L^6}^4\|\bfW\|^2 +
            \frac{\prN}{2}\|\del\bfW\|^2 
            +\|\del S\|^2.
\end{align*}
Next, for $II$ in \eqref{eq:DA:est:6}, we have from H\"older's
inequality, the Sobolev embedding $L^6(\Omega)\imb H^1$, Gagliardo-Nirenber interpolation inequality, and Young's inequality that
\begin{align*}
    II
    &\leq 2\| \bfW\|_{L^6}   \|\del\bfU\| \|\bfW\|_{L^3} 
    \leq C\|\del\bfW\|^{3/2}\|\bfW\|^{1/2} \|\del\bfU\|\notag\\
    &\leq  \frac{\prN}{2}\|\del\bfW\|^2+ \frac{C}{\prN^3}\|\del\bfU\|^4\|\bfW\|^2.
\end{align*}
For $III$, we use the Cauchy-Schwarz inequality to obtain
\begin{align}
III \leq 2\raN^2 \prN\|S\|^2 + \frac{\prN}{2}\|\bfW\|^2.
\label{eq:DA:est:8}
\end{align}

Combining the bounds for $I$-$III$, we have
\begin{align*}
 \ddt(  \|\bfW\|^2 +  \|S\|^2) + \prN \|\bfW\|^2 + 2 (\mu - \raN^2 \prN)\|S\|^2
 \leq
 \left(  \frac{C}{\prN^3} \|\del\bfU\|^4 + \frac{C}{\prN}\|T\|_{L^6}^4\right) \|\bfW\|^2.
\end{align*}
Thus, by the second condition in \eqref{3d:lam:hyp}, it follows that
\begin{align*}
 \ddt(  \|\bfW\|^2+  \|S\|^2) + 
  \left( \prN - \frac{C}{\prN^3} \|\del\bfU\|^4 - \frac{C}{\prN}\|T\|_{L^6}^4\right) 
  (\|\bfW\|^2  +\|S\|^2 )\leq  0.
\end{align*}
Therefore, by Gronwall's inequality the desired bound
\eqref{eq:f:p:syncon:ineq} now follows.
\end{proof}
\end{Lem}

\begin{Thm}\label{thm:finite}
Assume the conditions of Lemma \ref{lem:stab:3d}.  There exists a constant $C_1=C_1(\Om)>0$ such that if
\begin{align}\label{pr:large:3d}
	\prN\geq C_1\raN^5,
\end{align}
then
\begin{align}
  \|(\bfUT-\bfU)(t)\|+\|(\T-T)(t)\|  \leq C_2e^{-(\prN/4) t},
  \label{eq:finite:pr:sync}
\end{align}
holds for all $t\geq0$, where $C_2=\|\bfUT_0-\bfU_0\|+ \|\T_0-T_0\|$.
\begin{proof}
 By $(F3)$ of the Standing Hypotheses, it follows that
\begin{align}\label{est:stab2}
  \frac{C}{\prN^3}\int_0^t\|\del\bfU(s)\|^4ds
        \leq \frac{C\kap_1^4}{\prN^3}\raN^4 t.
\end{align}
By $(F4)$ and the Sobolev embedding,
  $H^1(\Omega)\imb L^6(\Omega)$, it follows that
\begin{align}\label{est:stab1}
    \frac{C}{\prN}\int_0^t \|T(s)\|_{L^6}^4ds 
             \leq \frac{C(\kap_1')^{6}}{\prN}\raN^{10}t.
\end{align}
Now upon combining \eqref{est:stab2}, \eqref{est:stab1}, it follows that there exists
$C=C(\Om)>0$ such that
\begin{align}
  \exp \left(C 
        \int_0^t\left( \frac{1}{\prN^3}\|\del\bfU(s)\|^4
                + \frac{1}{\prN}\|T(s)\|_{L^6}^4\right) 
  ds\right)
  \leq \exp\left(C\frac{\raN^4}{\prN}
  \left(\raN^6+\frac{1}{\prN^2}\right) t\right).\notag
\end{align}
Then, having chosen $\prN$ according to \eqref{pr:large:3d}, upon taking square roots, we deduce
\eqref{eq:finite:pr:sync} from and Lemma
\ref{lem:stab:3d}, as desired.
\end{proof}
\end{Thm}

\begin{Rmk}
Observe that \eqref{pr:large:3d} requires the Prandtl number to be very
large to ensure that synchronization occurs.  Indeed, for
$\raN \sim 10^7$ this condition indicates that $\Pr \gtrsim 10^{35}$
which would place the Prandtl number in a regime beyond the one that occurs for
the earth's mantle.  In addition, the lower bound on $\mu$ stated here
would imply that $\mu \gtrsim \Pr\raN^2$ which would yield a
very stiff problem numerically, particularly if $\Pr \sim \raN^5$
as indicated.  Gratefully, as seen in the next subsection, these
rigorous estimates are pessimistic, and synchronization is achieved
for much lower values of $Pr$ and $\mu$ than are indicated here.
\end{Rmk}

\subsection{Numerical Results}\label{section:numerical-results-finite} 

The numerical results are similar to those presented in Section
\ref{section:numerical-results-infinite}, i.e. the same spatial discretization and time-stepping algorithm are employed here, but now we also
consider variations in $\prN$.  In addition, for finite $\prN$ the velocity field is no longer slaved directly to the temperature field, requiring us to specify an initial condition for the velocity as well as the temperature.  This is done as in the infinite $\prN$ case by setting the initial state of the assimilated system to a low-order projection of the `true' system.  Simulations at higher Rayleigh numbers are initiated with flow fields from previous simulations at incrementally lower $\raN$.  To ensure that the algorithm is working, we first check that synchronization
still occurs for reasonable choices of $\raN$, $\mu$, and $N$ given a
finite $\prN$, say $\prN = 100$.  See Figure \ref{fig:basic-finite}.

\begin{figure}
\centering
\includegraphics[width=\textwidth]{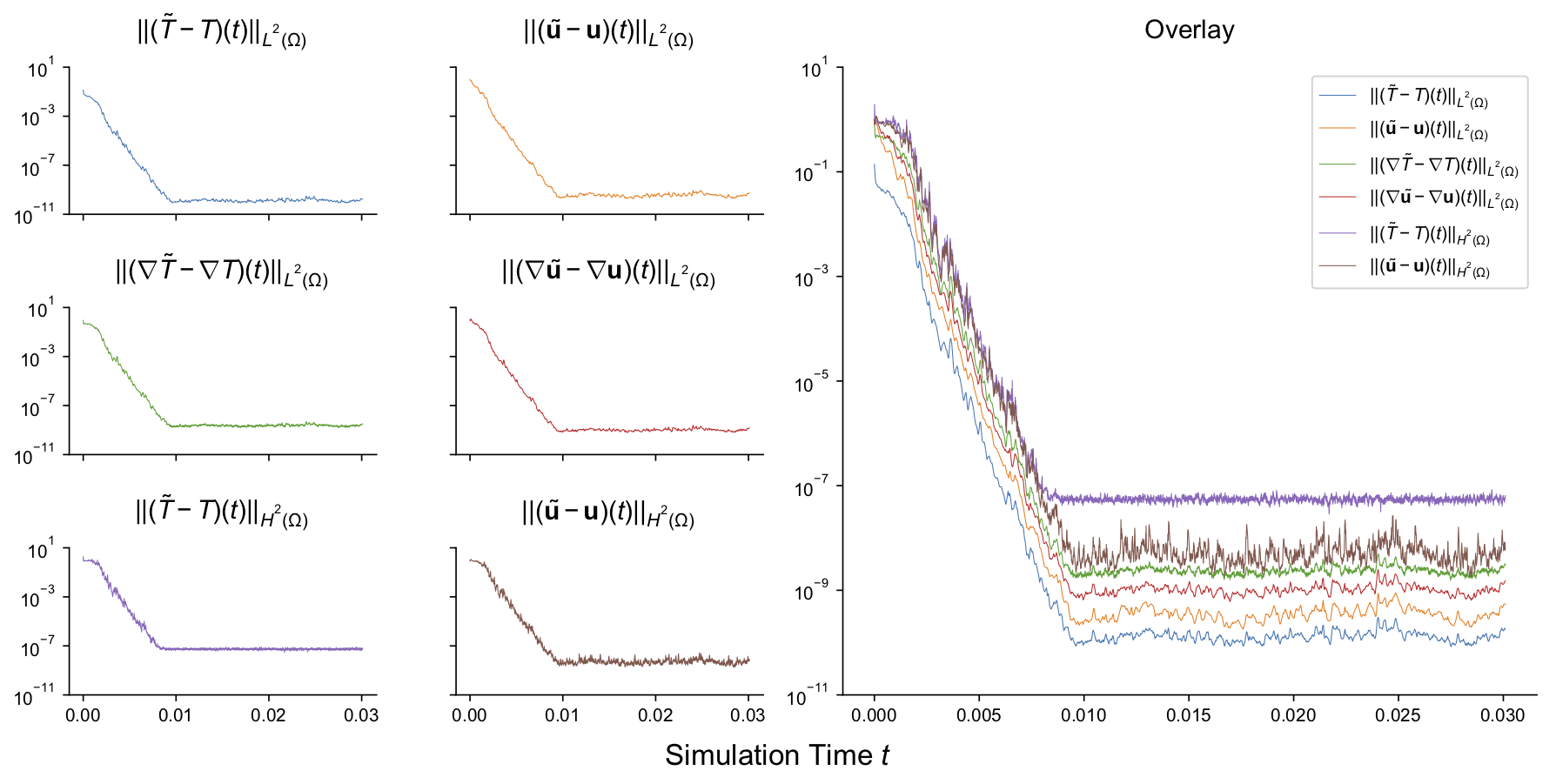}
\caption[Synchronization in various norms for
$\raN \approx 5.22\times 10^7$, $\mu=18,000$, $N=32$, and
$\prN = 100$.]{Synchronization in various norms for
  $\raN \approx 5.22\times 10^7$, $\mu=18,000$, $N=32$, and
  $\prN = 100$.  The temperature and velocity differences still
  decrease exponentially to zero, despite the fact that
  $\prN < \infty$.  However, the convergence is slow compared to the
  $\prN = \infty$ case (see Figure \ref{fig:vary-pr}).}
\label{fig:basic-finite}
\end{figure}

To simplify our exploration, for each value of $\raN$ listed in Tables
\ref{table:mu_values} and \ref{table:N_values}, we pick $\mu$ so that
the systems synchronize quickly.  Then, with $N = 32$, we run
simulations for logarithmically spaced values of $\prN$ from $1$ to
$100$.  Even with these larger-than-necessary choices of $\mu$,
convergence is lost for small enough $\prN$.  See Figure
\ref{fig:vary-pr} for a few additional examples and Table
\ref{table:pr_values} for the chosen $\mu$ and the lowest $\prN$ where
synchronization still occurs.

\begin{figure}
\centering
\begin{subfigure}{.49\textwidth}
    \centering
    \includegraphics[width=\textwidth]{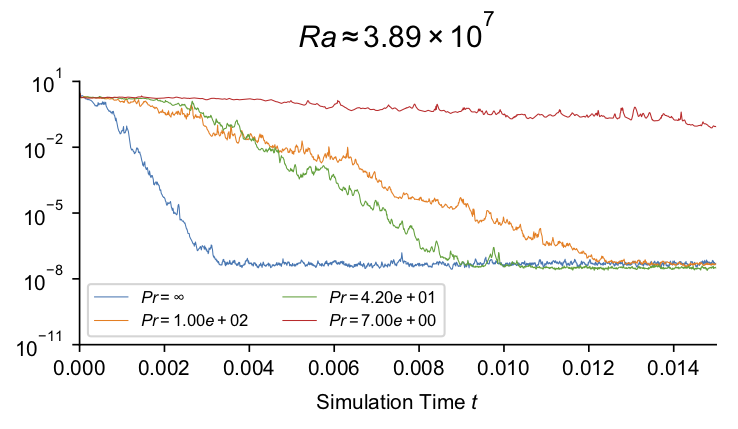}
\end{subfigure}
\begin{subfigure}{.49\textwidth}
    \centering
    \includegraphics[width=\textwidth]{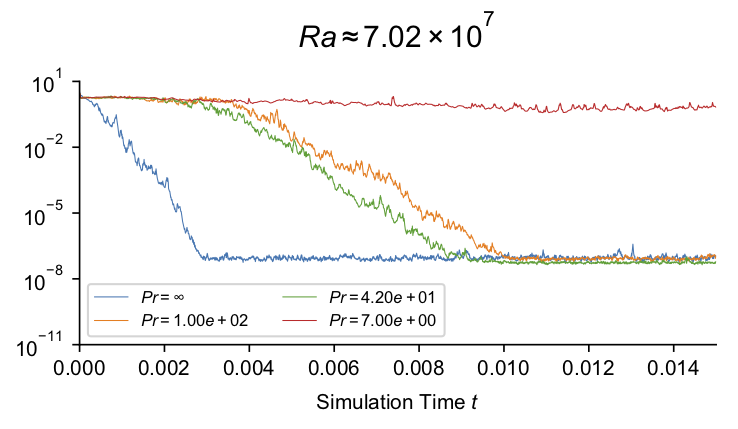}
\end{subfigure}
\caption[Synchronization with $N = 32$ and various $\prN$ for two
different values of $\raN$, with $\mu$ as given in Table
\ref{table:pr_values}.]{$\|(\bfUT-\bfU)(t)\|_{H^2} +
  \|(\T-T)(t)\|_{H^2}$
  with $N = 32$ and various $\prN$ for two different values of $\raN$,
  with $\mu$ as given in Table \ref{table:pr_values}.  The convergence
  is generally slower as $\prN$ decreases, until eventually there is
  no synchronization if $\prN$ is too small.}
\label{fig:vary-pr}
\end{figure}

\begin{table}
\centering
\begin{tabular}{r|r||c}
\multicolumn{1}{c|}{$\raN$} & \multicolumn{1}{c||}{$\mu$} & $\prN$ \\ \hline
$1.1937766 \times 10^7$ & $10,000$ & 75 \\
$1.6037187 \times 10^7$ & $11,000$ & 56 \\
$2.1544346 \times 10^7$ & $12,000$ & 56 \\
$2.8942661 \times 10^7$ & $13,000$ & 100 \\
$3.8881551 \times 10^7$ & $14,000$ & 42 \\
$5.2233450 \times 10^7$ & $18,000$ & 42 \\
$7.0170382 \times 10^7$ & $20,000$ & 42 \\
$9.4266845 \times 10^7$ & $25,000$ & 31 \\
$1.26638017\times 10^8$ & $32,000$ & 31 \\
$1.70125427\times 10^8$ & $40,000$ & 31
\end{tabular}
\caption[Minimal values of $\prN$ that result in synchronization within about 0.01 units of simulation time for the $\raN$ and $\mu$ in Table \ref{table:mu_values}, with $N = 32$.]{Minimal values of $\prN$ that result in synchronization within about 0.01 units of simulation time for the given $\raN$ and $\mu$ with $N = 32$.
The relationship here is much less precise than those described in Tables \ref{table:mu_values} and \ref{table:N_values} (for example, $\raN\approx 2.89 \times 10^7$ is a bit of an outlier), but it is interesting to note that lower $\raN$ seem to require larger $\prN$.
Whether or not this is an effect of $\raN$ increasing or $\mu$ increasing, however, is unclear.
}
\label{table:pr_values}
\end{table}

There are two main points to take away from these results.
\begin{itemize}
\item Finite-but-large Prandtl data assimilation through only
  temperature measurements is possible, though it is slower than the
  infinite Prandtl setting explored in Section
  \ref{section:infinite-prandtl-assimilation}.
\item The relationship between $\raN$, $\mu$, $N$, and $\prN$ remains
  unclear and would require further measurements to precisely
  quantify.  However, the fact that the assimilation works at all with
  reasonably small $\prN$ indicates that the inequality constraints in
  Theorem \ref{thm:finite} are strongly overstated, or that the constants
  in \eqref{pr:large:3d} are extremely small.
  \item While the estimates obtained in Theorem \ref{thm:finite} are clearly pessimistic it is interesting to note that the qualitative behavior is as expected.  That is, the finite Prandtl system will synchronize so long as $\prN$ is sufficiently large.  This is not unexpected as the limit of $\prN\rightarrow 0$ will be dominated by inertial effects in which the temperature and velocity field are nearly decoupled so we would anticipate that temperature observations alone will not suffice to recover the full flow field.
\end{itemize}

\section{A Scenario of Model Error} 
\label{section:hybrid-prandtl-assimilation}

Finally, we address a realistic scenario of assimilating observables that
correspond to solutions of the finite, but large Prandtl number system
\eqref{nd:Bou}, \req{nd:bc} into the ``incorrect," albeit more computationally tractable,
infinite Prandtl data assimilation system
\eqref{system:nudge-infinite}.  Thus, we suppose $(\bfU,T)$ satisfy $(F1)-(F4)$ and simply choose a suitable
$\T(\bfX,0) = \T_0(\bfX)$ for \eqref{system:nudge-infinite} since the corresponding initial velocity is enslaved by the temperature
evolution for the later nudging equation. Before we perform the analysis for the error estimates, we state the well-posedness result corresponding to the nudged equation, whose proof follows along similar lines to that of Theorem \ref{thm:nudge:infinite}.

\begin{Thm}\label{thm:nudge:hybrid}
Let $\mu>0$, and $\{\lambda_n\}_{n=1}^\infty$ be as in \req{eq:eigenpairs}. Let $(\bfU,T)$ of \eqref{nd:Bou}-\req{nd:bc} satisfying $(F1)-(F4)$. Let $\T_0\in L^2(\Om)$ a.e. in $\Om$ such that $\til{T}_0$ is
  a.e. $L$-periodic in $x_1,x_{d-1}$ with
  ${\T}_0|_{x_d=0}=0,{\T}_0|_{x_d=1}=1$ (in the sense of trace). Suppose $N>0$ satisfies
$\frac{1}{4}\lambda_N \ge \mu$.
Then there exists a unique $(\bfUT,\til{T})$ satisfying \req{system:nudge-infinite} in the weak sense  such that
  \begin{align*}
    \bfUT \in L^\infty(0,\tau;\mathcal{W}),
    \quad 
    \T\in L^\infty(0,\tau;L^2(\Om))\cap 
         L^2(0,\tau;H^1(\Om))\cap C_w([0,\tau],L^2(\Om)),
  \end{align*}
  for all $\tau>0$. 
\end{Thm}

\subsection{Error estimates} 

Since our data assimilation equation \eqref{system:nudge-infinite}
does not correspond to the true evolution of the observables
\eqref{nd:Bou}, \req{nd:bc}, we do not expect to obtain an exact
synchronization.  Instead, we derive estimates that quantify the
maximal error possible, and which will vanish as the
Prandtl number is taken increasingly large.

\begin{Thm}\label{thm:approx:infty}
Let $N,\mu>0$ satisfy $\frac{1}4\lam_N\geq \mu$ and $\til{T}_0$ be given as in Theorem \ref{thm:nudge:hybrid}. Let $(\bfUT,\T)$ be the corresponding unique solution to \req{system:nudge-infinite} guaranteed by Theorem \ref{thm:nudge:hybrid}. There exists a constant $C_0=C_0(\Om)>0$ such
  that if
  \begin{align}\label{approx:hyb:hyp} 
    \mu\geq4C_0\left(\raN^8+\raN^2\right),
  \end{align}
  then there exist positive
  constants $C_1=\|\T_0-T_0\|,C_2=C_2(\Om), C_3=C_3(\Om)$ such that
  \begin{align}\label{error:hyb:infty}
    \raN^{-1}\|\bfUT(t)-\bfU(s)\|_{H^2}+\|\T(t)-T(t)\|
    \leq C_1e^{-(\mu/2) t}+ \frac{C_2}{\prN}\frac{\left(\raN^{7/2}+\raN\right)}{\mu^{1/2}}+\frac{C_3}{\prN}\left(\raN^{5/2}+\raN^{7/4}\right),
  \end{align}
 for all $t\geq0$.
\end{Thm}

\begin{proof}
Let $S=\tilde{T}-T$, $\bfW=\bfUT-\bfU$, $q=\til{p}-p$. Then
\begin{align}
	-\Delta \bfW+\del q &=\raN \bfe_3 S+\frac{1}{\prN}[\bdy_t\bfU+(\bfU\cdotp\del)\bfU],\label{eqn:w:hybrid}\\
	  \bdy_tS+\bfW\cdotp\del S+\bfU\cdotp\del S
    &=-\De S-\mu P_NS-\bfW\cdotp\del T\label{eqn:S:hybrid}
\end{align}
Upon taking the $L^2$-inner product of $\bfW, S$ with \req{eqn:w:hybrid}, \req{eqn:S:hybrid}, respectively, then adding the consequent relations, we obtain
	\begin{align}
		\frac{1}2\frac{d}{dt}\Sob{S}{}^2&+\Sob{\del S}{}^2+\mu\Sob{S}{}^2+\Sob{\del\bfW}{}^2\notag\\
		&=\mu\Sob{Q_NS}{}^2-\lb\bfW\cdotp\del T,S\rb+\raN\lb e_3S,\bfW\rb+ \frac{1}{\prN}\lb\bdy_t\bfU+(\bfU\cdotp\del)\bfU,\bfW\rb\notag\notag\\
		&=I+II+III+IV\notag.
	\end{align}
Upon applying  H\"older's inequality, the Poincar\'e inequality, the Sobolev embedding theorem, Ladyzhenskaya's inequality, and Theorem \ref{prop:3d:Bou} $(ii)$, we derive
	\begin{align}
		|I|&\leq\frac{\mu}{\lam_N}\Sob{\del S}{}^2\notag\\
		|II|&\leq\Sob{\bfW}{L^3}\Sob{\del T}{L^6}\Sob{S}{}\notag\\
			&\leq C\kap_2'\raN^8\Sob{S}{}^2+\frac{1}{8}\Sob{\del\bfW}{}^2\notag\\
		|III|&\leq\raN\Sob{S}{}\Sob{\bfW}{}\notag\\
			&\leq C\raN^2\Sob{S}{}^2+\frac{1}{8}\Sob{\del\bfW}{}^2\notag\\
		|IV|&\leq \frac{1}{\prN}\left(\Sob{\bdy_t\bfU}{}\Sob{\bfW}{}+\Sob{\bfU}{L^4}^2\Sob{\del\bfW}{}\right)\notag\\
		&\leq\frac{C}{\prN}\left(\raN^{7/2}\Sob{\bfW}{}+\raN^{2}\Sob{\del\bfW}{}\right)\notag\\
		&\leq \frac{C}{\prN^2}(\raN^7+\raN^{2})+\frac{1}{8}\Sob{\del\bfW}{}^2\notag.
	\end{align}
Upon combining estimates for $I-IV$  with the conditions in \req{approx:hyb:hyp} and the fact that $\frac{1}4\lam_N\geq\mu$, it follows that
	\begin{align}\notag
	\frac{d}{dt}\Sob{S}{}^2+\mu\Sob{S}{}^2+\Sob{\del\bfW}{}^2\leq \frac{C}{\prN^2}(\raN^7+\raN^{2}).
	\end{align}
Hence, by Gronwall's inequality, we arrive at
	\begin{align}\label{hyb:S}
		\Sob{S(t)}{}^2\leq e^{-\mu t}\Sob{S_0}{}^2+\mu^{-1}\frac{C}{\prN^2}(\raN^7+\raN^{2})(1-e^{-\mu t}).
	\end{align}
Lastly, by Lemma \ref{lem:stokes} and $(F3)$ of the Standing Hypotheses, we have
	\begin{align}\label{hyb:w}
		\raN^{-1}\Sob{\bfW(t)}{H^2}\leq \Sob{S(t)}{}+\frac{C}{\prN}\left(\raN^{5/2}+\raN^{7/4}\right).
	\end{align}
We take the square root of \req{hyb:S} and add the result to \req{hyb:w} to complete the proof.

\end{proof}

\subsection{Numerical Results}\label{section:numerical-results-hybrid} 

To numerically verify Theorem \ref{thm:approx:infty} in a way that is
consistent with the numerical simulations corresponding to Theorems
\ref{thm:infty} and \ref{thm:finite}, we compare $(\bfUT,\T)$ to
$(\bfU,T)$.  To begin, consider a simulation with
$\raN\approx 5.22\times 10^7$, $\mu=18,000$ $N=32$, and $\prN = 100$.
For the finite Prandtl model in Section
\ref{section:finite-prandtl-assimilation}, this set of parameters
results in synchronization (see Figure \ref{fig:basic-finite}).  In
this situation, however, the synchronization appears to be limited by
the $O(\Pr^{-1})$ error from the $\Pr=\infty$ model to the
$\Pr< \infty$ reality.

\begin{figure}
\centering
\includegraphics[width=\textwidth]{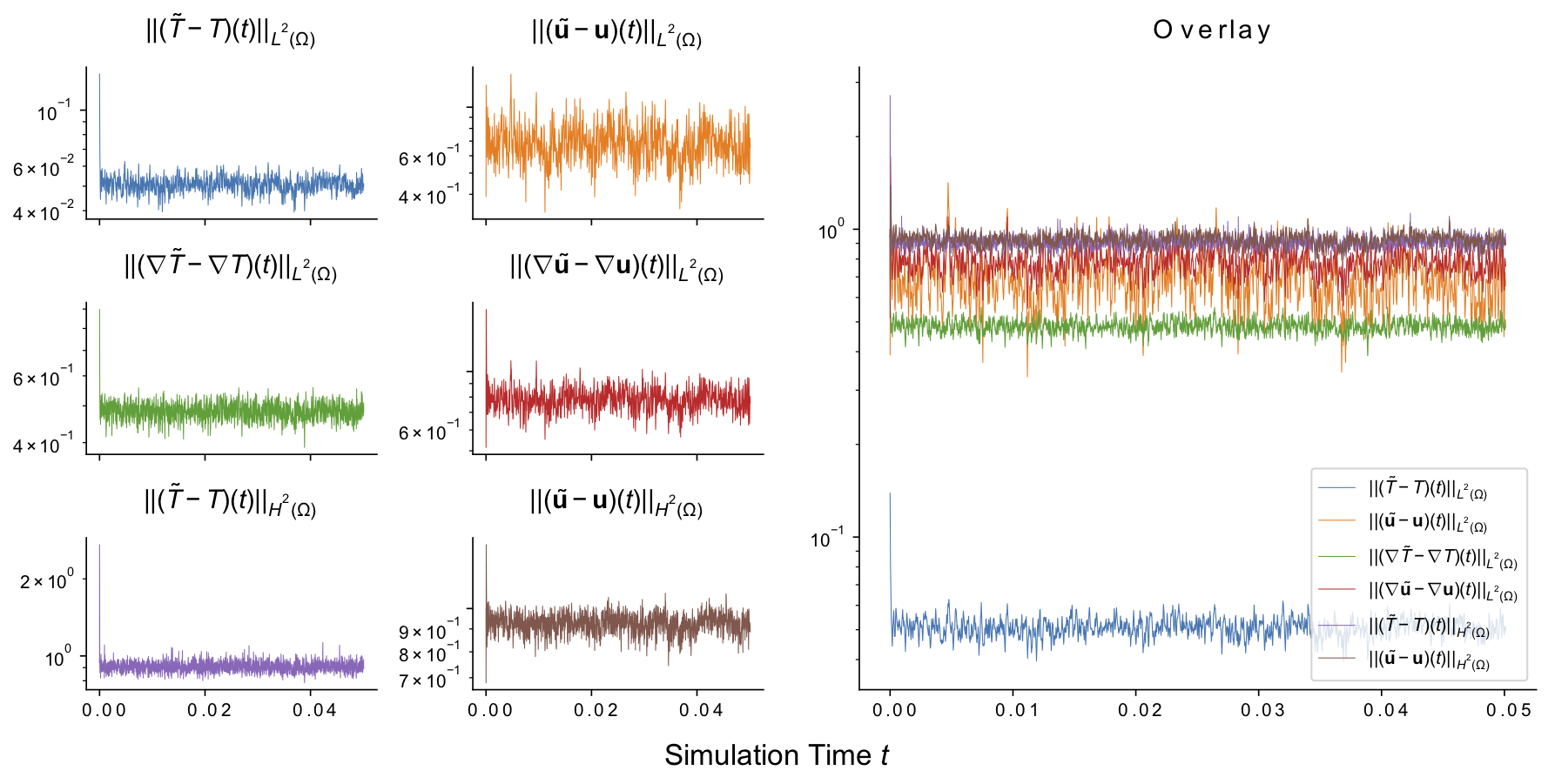}
\caption{Hybrid assimilation with $\raN\approx 5.22\times 10^7$,
  $\mu=18,000$, $N=32$, and $\prN = 100$.  The temperature and
  velocity differences remain almost constant, with no hint of
  convergence.  Note the difference in vertical axes relative to
  Figures exhibited previously in Sections
  \ref{section:numerical-results-infinite},
  \ref{section:numerical-results-finite}, above.}
\label{fig:basic-hybrid}
\end{figure}

This apparent lack of convergence is expected, however, since Theorem
\ref{thm:approx:infty} only guarantees that the error between $(\bfUT,\T)$
and $(\bfU,T)$ decreases to $O(\prN^{-1})$ as time
increases.  Using the same set of parameters, but with larger and
larger $\prN$, results in tighter and tighter synchronization.  To
more carefully match the results to the statement of Theorem
\ref{thm:approx:infty}, we calculate
$\raN^{-1}\|(\bfUT - \bfU)(t)\|_{H^2}$ +
$\|(\T - T)(t)\|$ at each simulation time $t$.  See
Figure \ref{fig:hybrid-works}.

\begin{figure}
\centering
\begin{subfigure}{.49\textwidth}
    \centering
    \includegraphics[width=\textwidth]{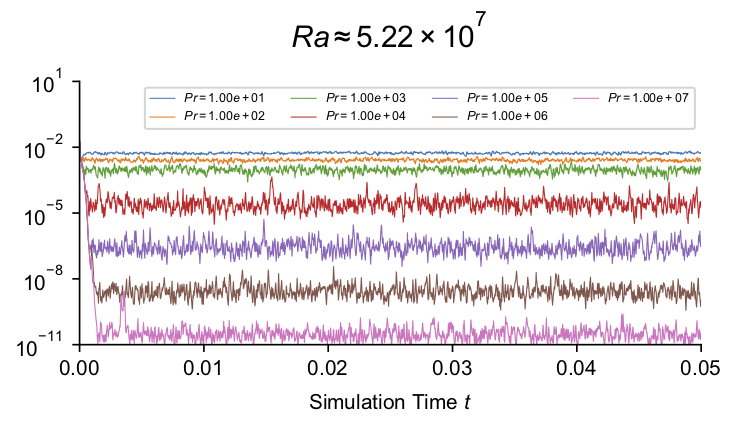}
\end{subfigure}
\begin{subfigure}{.49\textwidth}
    \centering
    \includegraphics[width=\textwidth]{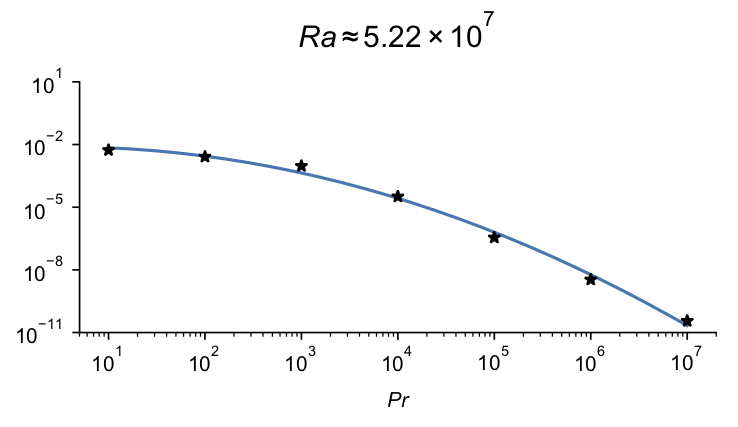}
\end{subfigure}
\caption{$\raN^{-1}\|(\bfUT - \bfU)(t)\|_{H^2}$ +
  $\|(\T - T)(t)\|$ with $\raN\approx 5.22\times 10^7$, $\mu=18,000$
  $N=32$, and various $\prN$.  As $\prN$ increases, the minimal error
  between $(\bfUT,\T)$ and $(\bfU,T)$ decreases.  The plot on the
  right describes the relationship between $\prN$ and the error with a
  quadratic least-squares fit.}
\label{fig:hybrid-works}
\end{figure}

While Figure \ref{fig:hybrid-works} is encouraging, it also highlights
a weakness in the data assimilation scheme unique to this hybrid
setting.  In both the previous settings considered here, the rigorous
estimates were pessimistic particularly for the large but finite
Prandtl case.  It appears that the estimates provided here for this
hybrid setting are far more true to practice, i.e. while the error is
not linear in $\Pr$ as seen in Figure \ref{fig:hybrid-works} it is
certainly dominated by the effects of the Prandtl number.  This
suggests that the practical success of these types of data
assimilation schemes is highly dependent on the data coming from the
same model as the simulated system, an unrealistically stringent
restriction.  Recall that the Boussinesq approximation is effectively
a ``zeroth order'' approximation for the mantle, meaning the true
physical system has several complicated secondary effects (some of
which are unknown) that are not included in the model.  The lack of
numerical synchronization in such a simple setting suggests that data
assimilation may not be as adaptable to settings where the exact model
is not known.

\section{Conclusions and Outlook} 
\label{section:conclusion}

In Section \ref{section:infinite-prandtl-assimilation}, we examined a
data assimilation scheme for the Rayleigh-B\'enard system with
$\prN = \infty$ and showed rigorously that synchronization occurs
between the data and assimilating equations under certain conditions
on the relaxation parameter $\mu$ and the number of projected modes
$N$ relative to the Rayleigh number $\raN$ when measurements of the temperature only are observed.  That is, as long as there
is enough data (i.e., $N$ is not too small), $\mu$ can be chosen large
enough to guarantee synchronization.  Though this is a satisfying
theoretical result, the numerical experiments in Section
\ref{section:numerical-results-infinite} show that synchronization
often occurs under much weaker conditions on $\mu$ and $N$ than
Theorem \ref{thm:infty} requires.  In particular, the inequality
conditions on $\mu$ is shown to be at least an order of magnitude away
from being sharp.  In addition, the numerical results also demonstrate
situations in which synchronization fails, namely when $\mu$ and/or
$N$ are not large enough.

Section \ref{section:finite-prandtl-assimilation} shows that
synchronization in the temperature measurements only and at
finite $\prN$ is also possible, although the rate of convergence is
slower than with $\prN = \infty$ and the relationship between $\raN$,
$\prN$, $\mu$, and $N$ needed to achieve synchronization remains
somewhat ambiguous.  As in the $\prN = \infty$ case, the conditions
imposed on $\mu$ appear quite pessimistic when compared to numerical
experiments.

Finally, when the true values are taken from $\prN < \infty$
simulations, but the assimilating equations use $\prN = \infty$, the
synchronization is highly dependent on $\prN$, as predicted by the
rigorous bounds.  This hybrid setting illustrates that the difference
between the two systems is dominating the error, rather than the
dynamical error in the synchronization process.  Although the
numerical simulations agree well with the rigorous predictions in this
setting, they do indicate a pessimistic outlook for additional
settings wherein the exact evolution of the dynamics for a data
assimilation system of this type is unknown.  In particular as noted
above, we have omitted several details in our model of mantle
convection that play a vital role in the evolution and may have an
effect similar to the difference between finite and infinite $\Pr$.
To investigate this further, data assimilation applied to the
internally heated convective setting (see
\cite{Go2016,WhDo2011b,WhDo2012} for example), and possibly the
anelastic or compressible convective systems \cite{Kingetal2010} will
be explored.

The current consideration of the difference between the infinite and
near-infinite Prandtl number convective systems lends itself to
further investigations wherein the assimilating model is different
from the physical system wherefrom the observations are obtained.  For
example, one might consider the effects of imprecisely defined
boundary conditions, i.e. what if the observations were obtained from
a convective simulation in which the velocity satisfied a Navier-slip
condition, but the nudged system was modeled with a no-slip condition?
Other variations in the model itself might include slight variations
in the geometry between the two systems, and additional terms in the
equations themselves such as internal heating mentioned above.  The
rub of the matter is that data assimilation techniques, if they are
meant to apply to physical settings such as weather, climate, and
investigations of the earth's mantle, must consider the fallibility of
the model they are relying on, that is do variations in the underlying
model itself allow for synchronization of the model with the observed
truth?

\section*{Acknowledgments} 

This paper was initiated by conversations between AF and NEGH (and
eventually JPW) which took place during a workshop at the Institute
for Pure and Applied Mathematics (IPAM), for which all the authors are
indebted to IPAM, the National Science Foundation (NSF), and
organizers of the program on Mathematics of Turbulence in Fall
2014. NEGH would like to acknowledge the grants NSF-DMS-1816551,
NSF-DMS-1313272 and Simons Foundation 515990 which supported this
work. SM and JPW acknowledge computational resources and support from
the Fulton Supercomputing Laboratory at Brigham Young University where
all the simulations were performed, as well as extensive feedback from
the dedalus developers (see \url{http://dedalus-project.org/}).  JPW
also acknowledges the hospitality and generosity of the Mathematics
Department at Tulane University where some of the finishing touches on
this work were carried out.

\addcontentsline{toc}{section}{References}

\begin{footnotesize}
\bibliographystyle{alpha}
\bibliography{bib}
\end{footnotesize}

\newpage

\vspace{.3in}
\begin{multicols}{2}
\noindent
Aseel Farhat\\
{\footnotesize Department of Mathematics\\
Florida State University\\
Web: \url{https://aseelfarhat842.wixsite.com/af7py}\\
Email: \href{mailto:afarhat@fsu.edu}{\nolinkurl{afarhat@fsu.edu}}} \\[.25cm]

\noindent 
Nathan Glatt-Holtz\\ {\footnotesize
Department of Mathematics\\
Tulane University\\
Web: \url{http://www.math.tulane.edu/~negh/}\\
Email: \href{mailto:negh@tulane.edu}{\nolinkurl{negh@tulane.edu}}} \\[.2cm]

 \noindent Vincent Martinez\\
{\footnotesize
Department of Mathematics and Statistics\\
Hunter College, CUNY\\
Web: \url{http://wpmin.hunter.cuny.edu/details/6494}\\
Email: \href{mailto:vincent.martinez@hunter.cuny.edu}{\nolinkurl{vincent.martinez@hunter.cuny.edu}}} \\[.2cm]
\columnbreak

\noindent Shane McQuarrie\\
{\footnotesize
  The Institute for Computational Engineering Sciences\\
  University of Texas, Austin\\
  Email: \href{mailto:shanemcq@utexas.edu}{\nolinkurl{shanemcq@utexas.edu}}} \\[.2cm]

 \noindent Jared Whitehead\\
{\footnotesize
Mathematics Department\\
Brigham Young University\\
Web: \url{https://math.byu.edu/~whitehead/}\\
Email: \href{mailto:whitehead@mathematics.byu.edu}{\nolinkurl{whitehead@mathematics.byu.edu}}} \\[.2cm]
 \end{multicols}

\end{document}